\documentclass[a4paper,11pt]{article}

\usepackage{jheppub} 

\usepackage{comment}
\usepackage[T1]{fontenc} 
\usepackage[utf8]{inputenc}
\usepackage{amsmath, amssymb, graphics, graphicx}
\usepackage[dvipsnames]{xcolor}
\usepackage{longtable}
\usepackage{subfigure}
\usepackage{multirow}
\usepackage{slashed}
\usepackage{dsfont}
\usepackage{appendix}
\usepackage[shortlabels]{enumitem}
\usepackage{amsthm}
\usepackage[final]{showlabels}
\usepackage{soul}

\usepackage{tikz-cd}
\usetikzlibrary{cd}

\newcommand{\lr}{\left(}
\newcommand{\rr}{\right)}
\newcommand{\ls}{\left[}
\newcommand{\rs}{\right]}
\newcommand{\lc}{\left\{}
\newcommand{\rc}{\right\}}

\newtheorem{theorem}{Theorem}

\def\rmi{\mathrm{i}}
\def\rme{\mathrm{e}}
\def\rmd{\mathrm{d}}

\definecolor{darkgreen}{rgb}{0.0, 0.5, 0.0}

\preprint{RBI-ThPhys-2020-14}

\title{\boldmath Non-Relativistic Supersymmetry\\ on Curved Three-Manifolds}

\author[a,b]{E.~A.~Bergshoeff,}
\author[c]{A.~Chatzistavrakidis,}
\author[a]{J.~Lahnsteiner,}
\author[a]{L.~Romano,}
\author[d]{J.~Rosseel}

\affiliation[a]{Van Swinderen Institute, University of Groningen,\\Nijenborgh 4, 9747 AG Groningen, The Netherlands}
\affiliation[b]{Erwin Schr\"odinger International Institute for Mathematics and Physics,\\ Boltzmanngasse 9, A-1090, Vienna, Austria}
\affiliation[c]{Division of Theoretical Physics, Rudjer Bo\v{s}kovi\'c Institute,\\Bijeni\v{c}ka 54, 10000 Zagreb, Croatia}
\affiliation[d]{Faculty of Physics, University of Vienna,\\Boltzmanngasse 5, A-1090, Vienna, Austria}

\emailAdd{e.a.bergshoeff[at]rug.nl}
\emailAdd{Athanasios.Chatzistavrakidis[at]irb.hr}
\emailAdd{j.m.lahnsteiner[at]outlook.com}
\emailAdd{lucaromano2607[at]gmail.com}
\emailAdd{jan.rosseel[at]univie.ac.at}

  \def\g{\gamma}  \def\d{\delta} 
\def\e{\epsilon} 
   
\def\l{\lambda}  \def\m{\mu}
\def\n{\nu}    
  \def\t{\tau} 
  \def\z{\zeta}

\newcommand{\be}{\begin{equation}}
	\newcommand{\ee}{\end{equation}}
\newcommand{\sfrac}[2]{{\textstyle\frac{#1}{#2}}}
\def\nn{\nonumber}
\def \bea{\begin{eqnarray}} 
	\def\eea{\end{eqnarray}}
\def\bse{\begin{subequations}}	
	\def\ese{\end{subequations}}
\def\bal{\begin{align}} 
	\def\eal{\end{align}}

\def\bi{\begin{itemize}} 
	\def\ei{\end{itemize}}

\abstract{We construct explicit examples of non-relativistic supersymmetric field theories on curved Newton-Cartan three-manifolds. These results are obtained by performing a null reduction of four-dimensional supersymmetric field theories on Lorentzian manifolds and the Killing spinor equations that their supersymmetry parameters obey. This gives rise to a set of algebraic and
differential Killing spinor equations that are obeyed by the supersymmetry parameters of the resulting three-dimensional non-relativistic field theories. We derive necessary and sufficient conditions that determine whether a Newton-Cartan background admits non-trivial solutions of these Killing spinor equations. Two classes of examples of Newton-Cartan backgrounds that obey these conditions are discussed. The first class is characterised by an integrable foliation, corresponding to so-called twistless torsional geometries, and includes manifolds whose spatial slices are isomorphic to the Poincar\'e disc. The second class of examples has a non-integrable foliation structure and corresponds to contact manifolds.}

\begin{document}
\maketitle
\flushbottom

\section{Introduction}

Recent years have seen a lot of activity in the use of localization techniques to study non-perturbative aspects of supersymmetric Quantum Field Theories (susy QFTs). Following the work of \cite{Pestun:2007rz,Kapustin:2009kz}, this has for instance led to the calculation of exact partition functions of susy QFTs, defined on curved backgrounds that admit one or more Killing spinors. These Killing spinors then serve as parameters of the supersymmetry transformation rules that leave susy QFTs on the considered backgrounds invariant. The Lagrangian and transformation rules of susy QFTs on curved backgrounds generically contain various terms, in which the matter fields are non-minimally coupled to the background metric. Although these terms can in principle be obtained by applying the Noether procedure to minimally coupled theories, constructing susy QFTs in this way tends to be rather cumbersome in practice. A less involved and more insightful way to obtain susy QFTs on curved backgrounds was developed by Festuccia and Seiberg in \cite{Festuccia:2011ws} and consists of applying a rigid decoupling limit to matter field theories, that are coupled to off-shell supergravity. Building on this result, a better geometric understanding of the backgrounds on which susy QFTs can be defined, as well as further applications of supersymmetric localization techniques, have been obtained (see \cite{Pestun:2016zxk} for a review).

As mentioned above, the rigid supersymmetry parameters of susy QFTs on curved backgrounds are determined as solutions of Killing spinor equations. In \cite{Festuccia:2011ws}, these Killing spinor equations are obtained by setting the supersymmetry transformations of the fermions of the off-shell supergravity multiplet equal to zero. The equations thus obtained take the schematic form\footnote{This assumes that the gravitini are the only fermionic fields in the supergravity multiplet. In case the supergravity multiplet contains extra fermionic fields, these Killing spinor equations have to be supplemented with algebraic ones. See also the last two paragraphs of section 3.3 for comments on the relevance of such algebraic Killing spinor equations to non-relativistic supersymmetry discussed in this paper.}
\begin{equation}
    D_M\epsilon + \big(\mathcal{B}\Gamma\big)_M \epsilon =0\,. \label{eq:gKSE}
\end{equation}
Here, $\epsilon$ is the supersymmetry parameter, $D_M$ is a covariant spinor derivative and $\big(\mathcal{B}\Gamma\big)_M$ denotes a background one-form, that is matrix-valued in spinor space and that depends on gamma matrices as well as the bosonic  fields of the off-shell supergravity multiplet. Classifying classical backgrounds, on which susy QFTs can be formulated, then involves classifying the bosonic off-shell supergravity field configurations for which the equations \eqref{eq:gKSE} have non-trivial solutions for $\epsilon$. In particular, when restricting to field configurations for which only the metric field is non-trivial, eq.~\eqref{eq:gKSE} reduces to
\begin{equation} \label{eq:CovCons}
     D_M\epsilon = 0 \,.
\end{equation}
Under this restriction, susy QFTs can therefore only be defined on backgrounds that admit one or more covariantly constant spinors. Demanding the existence of a covariantly constant spinor constrains the geometry of a background to be Ricci-flat. This thus singles out tori $T^4$ and $K3$ surfaces, when restricting to compact Euclidean four-manifolds.

In order to obtain more general manifolds on which susy QFTs can be defined, one needs to consider off-shell supergravity backgrounds in which (e.g.~auxiliary) fields other than the metric are turned on, such that eq.~\eqref{eq:gKSE}  admits non-trivial solutions for $\epsilon$. Once such backgrounds are found, one can consider off-shell matter-coupled supergravity theories on them and take the rigid limit that freezes out the fluctuations of the supergravity fields around their background values. Taking this limit in the Lagrangian and supersymmetry transformation rules then leads to the Lagrangian and transformation rules of susy QFTs in non-dynamical curved backgrounds. The background values of the auxiliary fields of the supergravity multiplet are responsible for the non-minimal couplings that are necessary to maintain supersymmetry on a curved manifold.

Most of the developments mentioned above are concerned with the Euclidean case. In the non-Euclidean case, the literature mainly deals with relativistic backgrounds, i.e.~manifolds that are equipped with a non-degenerate Lorentzian metric \cite{Liu:2012bi, Cassani:2012ri}. Recent developments in non-relativistic holography \cite{Son:2008ye,Kachru:2008yh,Balasubramanian:2008dm,Christensen:2013lma,Christensen:2013rfa,Taylor:2015glc, Hartnoll:2016apf} and effective field theory methods for strongly coupled condensed matter systems \cite{Son:2005rv,Hoyos:2011ez,Son:2013rqa,Abanov:2013woa,Gromov:2014vla,Gromov:2015fda,Geracie:2014nka,Geracie:2015dea,Geracie:2015xfa,Geracie:2016inm} have however led to a renewed interest in non-relativistic QFTs on curved backgrounds as well. There exist various notions of non-relativistic differential geometry among which Newton-Cartan geometry is the prime example \cite{Cartan}. Given the usefulness of susy QFTs on curved backgrounds in studying relativistic QFTs in the non-perturbative regime, it is natural to ask whether susy QFTs on non-trivial Newton-Cartan backgrounds can be of similar importance. In order to address this question, one first needs to formulate susy QFTs on curved Newton-Cartan space-times. This is the problem that we will address in this paper.\footnote{In this paper, we will consider susy QFTs whose multiplets in the flat case correspond to representations of the super-Bargmann algebra. It would be interesting to see whether our results can be extended to consider susy QFTs whose multiplets in the flat case are representations of other non-relativistic supersymmetry algebras, such as e.g. super-Lifshitz algebra \cite{Chapman:2015wha,Arav:2019tqm}.}

To construct explicit examples of non-relativistic susy QFTs on curved Newton-Cartan manifolds, one could in principle apply the technique of \cite{Festuccia:2011ws} to matter field theories coupled to non-relativistic off-shell supergravity. In this regard, it is useful to point out that currently not much is known about non-relativistic off-shell supergravity. The only non-relativistic supergravity multiplets considered so far are three-dimensional ones. The original three-dimensional Newton-Cartan supergravity theory of \cite{Andringa:2013mma} is on-shell in the sense that the supersymmetry algebra only closes upon imposition of extra constraints. Some of these constraints can be recognized as fermionic equations of motion, like in the case of relativistic on-shell supergravity, while other constraints are geometrical constraints that have no relativistic on-shell supergravity analog. Extensions of this on-shell theory have been constructed, in which the supergravity algebra is realized without having to impose fermionic equations of motion \cite{Bergshoeff:2015uaa,Bergshoeff:2015ija}. However, for all these multiplets, one still needs geometric constraints in order to close the underlying non-relativistic superalgebra on the fields. It is at present not clear whether there exists a multiplet for which the superalgebra closes without the use of any constraints and from which the previously mentioned multiplets could be obtained as specific truncations. In view of this, it is not clear whether analyzing the Killing spinor equations, that stem from the supersymmetry transformations of the fermionic fields of these multiplets, leads to the most general non-relativistic backgrounds on which non-relativistic susy QFTs can be defined. Indeed, the authors of
\cite{Knodel:2015byb} found that the class of allowed maximally supersymmetric and $\frac12$-BPS backgrounds for one specific non-relativistic supergravity multiplet (constructed in \cite{Bergshoeff:2015uaa}) is rather restricted.

In this paper we will follow a different strategy and obtain non-relativistic susy QFTs in three dimensions by performing a dimensional reduction of relativistic four-dimensional susy QFTs over a lightlike isometry--a so-called null reduction. This is reminiscent of how Newton-Cartan gravity in four dimensions can be obtained as a null reduction of Einstein gravity in five dimensions \cite{Julia:1994bs}. As shown in \cite{Liu:2012bi, Cassani:2012ri}, analysis of the Killing spinor equations, stemming from Old and New Minimal supergravity, implies that four-dimensional relativistic backgrounds on which susy QFTs can be formulated, possess a null Killing vector. It is this fact that we will exploit to obtain non-relativistic susy QFTs on curved backgrounds from four-dimensional relativistic ones. For simplicity, we will restrict ourselves in this paper to the null reduction of four-dimensional theories that are obtained as a rigid limit of matter field theories coupled to Old Minimal supergravity, leaving the New Minimal case for future work.

A particular feature of our null reduction approach is that it ultimately relies on four-dimensional relativistic results. Nevertheless, we will be able to extract some general lessons that we expect to hold when discussing generic non-relativistic supersymmetric backgrounds. We will in particular pay attention to the structure of the three-dimensional non-relativistic Killing spinor equations and see that consistency with local non-relativistic symmetries leads one to include algebraic equations in the set of Killing spinor equations. We will then use these non-relativistic Killing spinor equations to discuss three-dimensional non-relativistic supersymmetric backgrounds in an intrinsically three-dimensional manner. This analysis is technically simpler than the relativistic four-dimensional one. One could thus also advocate combining non-relativistic geometry (of a kind that is obtainable from null reduction) with suitable Killing spinor equations as an alternative way to obtain interesting relativistic supersymmetric backgrounds via dimensional oxidation along a lightlike isometry.

This paper is organized as follows. In section 2 we collect some known results about supersymmetry on Lorentzian four-manifolds, obtained as a rigid limit of matter-coupled Old Minimal supergravity. In section 3 we apply the null reduction to obtain three-dimensional non-relativistic susy QFTs on curved backgrounds together with the Killing spinor equations that their supersymmetry parameters should satisfy. In section 4, we investigate the conditions that various background fields have to satisfy in order for non-trivial solutions of the non-relativistic Killing spinor equations to exist. We also discuss two classes of explicit examples of three-dimensional non-relativistic backgrounds, on which supersymmetry can be defined. We end with a conclusions and outlook section. There are also three appendices. Appendix A summarizes the conventions used in this paper. Appendix B collects a few technical formulae that are needed to perform the null reduction discussed in section 3. Finally, appendix C discusses the integrability conditions for the non-relativistic Killing spinor equations, giving an alternative derivation of some of the results of section 4.

\section{Supersymmetry on Lorentzian Four-Manifolds} \label{sec:susy4d}

Relativistic susy QFTs on curved space-times can be obtained by taking a rigid limit of matter-coupled off-shell supergravity theories \cite{Festuccia:2011ws}. This procedure consists of choosing a non-trivial (i.e.~non-flat) classical\footnote{The classical nature of the background implies that the fermionic fields of the supergravity multiplet assume zero background values.} background for the metric and auxiliary fields of the off-shell supergravity multiplet and taking the limit in which the Planck mass $M_P$ is sent to infinity (after assigning suitable mass dimensions to the fields of the supergravity multiplet). The limit $M_P \rightarrow \infty$ decouples the fluctuations of the supergravity multiplet fields so that one is left with the matter multiplets coupled to the chosen classical background, via minimal and typically also non-minimal coupling terms. In order for the resulting field theory to be supersymmetric, the background fields should be such that the Killing spinor equations, obtained by setting the supersymmetry transformations of the fermionic supergravity multiplet fields equal to zero, admit non-trivial solutions for the supersymmetry parameters. Since one works with off-shell supergravity these Killing spinor equations are independent of the choice of matter fields, which greatly simplifies the search for possible curved backgrounds on which susy QFTs can be formulated.

This limit was discussed explicitly in \cite{Festuccia:2011ws} for the case of chiral matter coupled to $4d$, $\mathcal{N}=1$ Old Minimal supergravity \cite{Cremmer:1982wb,Cremmer:1982en} with a metric $g_{MN}$\,\footnote{Curved indices $M$, $N$, $\cdots$ are raised and lowered using the background metric $g_{MN}$.} and auxiliary fields $\left\{U,V_M\right\}$ as bosonic components, where $U$ is a complex scalar (with complex conjugate $\bar{U}$) and $V_M$ is a real vector. The fermionic field content of the Old Minimal supergravity multiplet consists of a (Majorana) gravitino $\psi_M$, that is zero in a classical background. We mainly follow the notation of  \cite{Cremmer:1982wb,Cremmer:1982en}\,\footnote{Note that this notation is different from the one used in \cite{Festuccia:2011ws}, leading to different prefactors compared to the results of \cite{Festuccia:2011ws}.} but restrict to just one chiral multiplet with components  $\{Z, \chi_L, H\}$, where $Z$ is a dynamical complex scalar, $\chi_L$ a left-handed Weyl fermion and $H$ an auxiliary complex scalar.\,\footnote{Their complex conjugate, anti-chiral counterparts will be denoted by $\{\bar{Z}, \chi_R, \bar{H}\}$. We will also often combine a left-handed spinor $\chi_L$ and a right-handed one $\chi_R$ into a Majorana spinor $\chi$, defined as $\chi = \chi_L + \chi_R$.} Taking the rigid limit of Old Minimal supergravity, one obtains the following Lagrangian for a supersymmetric field theory of a chiral multiplet in a curved four-dimensional background\,\cite{Festuccia:2011ws}: 
\begin{align}
  \label{eq:Lrigid}
  E^{-1}\mathcal{L} &= -\frac{E^{-1}}{3}Z \bar{Z}\mathcal{L}_{SG} - \partial_M Z \,\partial^M \bar{Z} - \bar{\chi}\left(\slashed{D}-\frac{\rmi}{6}\slashed{V}\Gamma_5\right)\chi + H \bar{H}\notag\\
  &\quad\,+\frac13\left(\bar{U}\,\bar{Z}\,H + U\,Z\,\bar{H}\right) + \frac{\rmi}{3}V^M\left(\bar{Z}\partial_M Z - Z\partial_M\bar{Z}\right)\\
  &\quad\,+\mathrm{Re}\left(W''\bar{\chi}_L\chi_L - W' H  - W U \right)\,,\notag
\end{align}
where
\begin{equation}
E^{-1}\mathcal{L}_{SG} = -\frac12\,R - \frac13\,U \bar{U} + \frac13\,V^M V_M\,.
 \end{equation}
In these equations, $E$ is the square root of minus the determinant of the metric, $R$ the background Ricci scalar and the Lorentz-covariant spinor derivative $D_M\chi$ is defined by 
\begin{equation}
  \label{eq:spinorcovder}
  D_M \chi = \left(\partial_M + \frac14 \Omega_M{}^{AB} \Gamma_{AB} \right) \chi\,,
\end{equation}
with $\Omega_M{}^{AB}$ the background spin connection.
The function $W = W(Z)$ depends holomorphically on $Z$ and is the superpotential of the theory. Its derivatives with respect to $Z$ are denoted by
\begin{equation}
  \label{eq:derscalnot}
  W' = \frac{\rmd W}{\rmd Z}\,,\qquad W^{\prime\prime} = \frac{\rmd^2 W}{\rmd Z^2}\,.
\end{equation}
The Lagrangian \eqref{eq:Lrigid} is then invariant under the following supersymmetry transformation rules
\begin{align}
  \label{eq:susytrafos}
  \delta Z &= \bar{\epsilon}_L \chi_L \,, \nonumber \\
  \delta \chi_L &= \frac12 \slashed{\partial} Z \epsilon_R + \frac12 H \epsilon_L \,, \nonumber \\
  \delta H &= \bar{\epsilon}_R \left(\slashed{D} - \frac{\rmi}{6} \slashed{V} \right) \chi_L - \frac{U}{3} \bar{\epsilon}_L\chi_L\,,
\end{align}
provided that
the rigid supersymmetry parameters $\epsilon_{L/R}$ that appear in \eqref{eq:susytrafos} are solutions of the following Killing spinor equations
\begin{align}
  \label{eq:Killspinoreqs}
 & D_M \epsilon_L + \frac{\rmi}{2} V_M \epsilon_L + \frac16 \bar{U} \Gamma_M \epsilon_R - \frac{\rmi}{6} \Gamma_M \slashed{V} \epsilon_L = 0 \,, \nonumber \\ & D_M \epsilon_R - \frac{\rmi}{2} V_M \epsilon_R + \frac16 U \Gamma_M \epsilon_L + \frac{\rmi}{6} \Gamma_M \slashed{V} \epsilon_R = 0 \,.
\end{align}
 The Killing spinor equations \eqref{eq:Killspinoreqs} are obtained by requiring that supersymmetry preserves the chosen classical background for the Old Minimal supergravity multiplet. Since the background value of the gravitino is zero, the only non-trivial conditions that arise from this requirement, are obtained by setting the gravitino supersymmetry transformation rule, evaluated on the background, equal to zero. This then leads to \eqref{eq:Killspinoreqs}.
 One can explicitly check that the Lagrangian \eqref{eq:Lrigid} is invariant under the transformation rules \eqref{eq:susytrafos} provided that the Killing spinor equations \eqref{eq:Killspinoreqs} hold. 
 
Requiring that the Killing spinor equations \eqref{eq:Killspinoreqs} have non-trivial solutions leads to constraints on the background geometry and the auxiliary fields $\{V_M$, $U$, $\bar{U}\}$. Before discussing this in more detail, it is worth pointing out that many results in the literature \cite{Dumitrescu:2012ha,Dumitrescu:2012at} are strictly speaking only valid for Euclidean backgrounds, while in this paper we are interested in Lorentzian backgrounds. The difference between the Euclidean and Lorentzian cases manifests itself in the reality conditions that are imposed on the background values of the auxiliary fields $V_M$ and $U$, $\bar{U}$. In the Euclidean case, the background value $V_M$ is allowed to be complex while $U$ and $\bar{U}$ are allowed to correspond to two independent complex background scalars. Likewise, the Killing spinors $\epsilon_L$ and $\epsilon_R$ are treated as two independent Weyl spinors and the equations \eqref{eq:Killspinoreqs} are independent. In contrast, for the Lorentzian case one has to impose that $V_M$ is real, that $\bar{U}$ is the complex conjugate of $U$ and that the spinors $\epsilon_L$, $\epsilon_R$ are chiral projections of a Majorana spinor $\epsilon=\epsilon_L + \epsilon_R$ and thus related via complex conjugation. This, in turn, implies that the equations \eqref{eq:Killspinoreqs} are not independent but instead are each other's complex conjugate. 

We may assume that the Lagrangian \eqref{eq:Lrigid} and the supersymmetry transformation rules \eqref{eq:susytrafos} hold for both the Euclidean and Lorentzian cases as long as we assume  that the auxiliary background fields and Killing spinors obey the appropriate reality conditions. The analysis of the Killing spinor equations \eqref{eq:Killspinoreqs} that determines the allowed supersymmetric backgrounds depends more subtly on the signature of the background space-time and on the ensuing reality properties of the auxiliary background fields. The Lorentzian case was previously discussed in \cite{Liu:2012bi, Cassani:2012ri}\,\footnote{See \cite{Festuccia:2011ws, Dumitrescu:2012at} for analogous results in Euclidean signature.} and we briefly summarize some important parts and results of this analysis below.

When solving the Killing spinor equations \eqref{eq:Killspinoreqs}, it suffices to look for solutions that are commuting Majorana spinors. Note that the physical fermions $\{\chi_L,\chi_R\}$ are anti-commuting and that consequently the parameters $\epsilon_L$, $\epsilon_R$ of the supersymmetry transformations \eqref{eq:susytrafos} should be anti-commuting as well. Once one has however obtained a basis $\{\zeta^{(i)} = \zeta^{(i)}_L + \zeta^{(i)}_R | i = 1,\cdots,n\}$ of $n$ commuting Majorana solutions of \eqref{eq:Killspinoreqs}, one can use linearity of \eqref{eq:Killspinoreqs} to construct generic supersymmetry parameters $\epsilon$ as linear combinations of the $\zeta^{(i)}$, with constant, real Grassmann variables as coefficients:
\begin{equation}
  \epsilon = \sum_{i=1}^n \theta_{(i)} \zeta^{(i)} \,, \qquad \mathrm{with} \qquad \theta_{(i)} \theta_{(j)} = - \theta_{(j)} \theta_{(i)} \,,\ \theta_{(i)}^* = \theta_{(i)} \,.
\end{equation}
In the following, we will use the Greek letter $\zeta$ to denote commuting solutions of Killing spinor equations, while the letter $\epsilon$ will be reserved for the associated anti-commuting supersymmetry parameters.

Assuming the existence of commuting solutions of \eqref{eq:Killspinoreqs}, one can derive geometric restrictions that should be obeyed by backgrounds on which susy QFTs can be defined. One important restriction on the allowed backgrounds is that they admit a null Killing vector. Indeed, the existence of a non-trivial commuting Killing spinor $\zeta = \zeta_L + \zeta_R$ allows one to define the following real vector
\begin{equation}
    K^M = \rmi\, \bar{\zeta}\,\Gamma^M\zeta = 2 \rmi\, \bar{\zeta}_L \Gamma^M \zeta_R\,,\qquad\mathrm{which~obeys}\qquad K_MK^M=0\, \label{eq:Kkilling}
  \end{equation}  
as a consequence of Fierz relations. Moreover, using the Killing spinor equations \eqref{eq:Killspinoreqs}, one can show that \cite{Liu:2012bi}
\begin{align}\label{eq:Killing}
    \nabla_{(M}K_{N)}=0\qquad\mathrm{and}\qquad K_{[M}\nabla_NK_{P]} = 2\,\epsilon_{MNP}{}^QK_Q \, K^S V_S\,,
\end{align}
where $V_S$ is the real auxiliary vector of the Old Minimal supergravity multiplet. We thus see that $K^M$ is a null Killing vector, whose associated one-form $K_M = g_{MN} K^N$ is non-integrable (i.e. $K_{[M} \partial_N K_{R]} \neq 0$), unless $K_MV^M=0$.

More generally, given a basis $\{\zeta^{(i)} | i = 1,\cdots, n\}$ of commuting solutions of \eqref{eq:Killspinoreqs}, one can show that the vectors
\begin{equation}
  K^M_{(ij)} = \rmi \bar{\zeta}^{(i)} \Gamma^M \zeta^{(j)}\,,
\end{equation}
are Killing vectors and thus correspond to isometries of the background. The generators of these isometries determine the anti-commutators of the supercharges of the rigid superalgebra that is preserved by the background. Let us illustrate this in case there is one commuting solution $\zeta = \zeta_L + \zeta_R$ of the Killing spinor equation \eqref{eq:Killspinoreqs}. Associated to this solution, one can construct the supercharge $Q(\zeta)$, that generates supersymmetry transformations \eqref{eq:susytrafos}, whose parameters are of the form $\epsilon= \theta \zeta$, where $\theta$ is a real, constant Grassmann variable
\begin{equation}
  \delta(\epsilon= \theta \zeta) = \theta Q(\zeta) \,.
\end{equation}
Calculating the commutator of two such supersymmetry transformations $\delta(\epsilon_{1,2} = \theta_{1,2} \zeta)$ (with $\theta_{1,2}$ two independent anti-commuting variables) on the fields $Z$, $H$ and $\chi$, one obtains
\begin{equation}
  \left[\delta(\epsilon_1), \delta(\epsilon_2) \right] = -\frac{\rmi}{2} \theta_2 \theta_1 \mathcal{L}_{_K} \,,
\end{equation}
where $\mathcal{L}_{_K}$ is the Lie-Lorentz derivative \cite{Figueroa-OFarrill:1999klq, Ortin:2015hya} along the Killing vector $K^M$ defined in \eqref{eq:Kkilling}. This Lie-Lorentz derivative acts as an ordinary Lie derivative on the scalar fields $Z$, $H$ and as
\begin{equation}\label{eq:covLie}
     \mathcal{L}_{_K} \chi = K^M D_M \chi - \frac14 \left(D_A K_B\right) \Gamma^{AB}\chi\,,
\end{equation}
on the spin-1/2 fermionic field $\chi$. One thus sees that, in case there is only one solution $\zeta$ to the Killing spinor equations \eqref{eq:Killspinoreqs}, the part of the preserved rigid superalgebra that involves the associated single supercharge $Q(\zeta)$ is given by
\begin{equation}\label{eq:NMalgebra}
     \left\{Q(\zeta),Q(\zeta)\right\} = -\frac{\rmi}{2}\mathcal{L}_{_K}\qquad\mathrm{and}\qquad\left[Q(\zeta),\,\mathcal{L}_{_K}\right]=0\,.
   \end{equation}
 Equations \eqref{eq:Kkilling} and \eqref{eq:Killing} show that a necessary condition for a background to allow for supersymmetry is the existence of a (globally defined) null Killing vector. Hence the set of product manifolds
 \begin{equation}
  \mathbb{R}^{1,1}\times\mathcal{M}_2\,,
  \end{equation}
  with $\mathcal{M}_2$ being an arbitrary two-manifold, provides a large class of candidate solutions. There are other known consistent backgrounds that do not fall into this class. Two such backgrounds, that preserve maximal supersymmetry, are given by AdS$_4$ and $\mathbb{R} \times \mathbb{S}^3$. The Euclidean versions of these backgrounds have been constructed in \cite{Festuccia:2011ws}. The AdS$_4$ case was also discussed for Lorentzian signature in \cite{Liu:2012bi}. The $\mathbb{R}\times \mathbb{S}^3$ background is an example where the one-form $K_M$ is not integrable. We will consider non-relativistic supersymmetric manifolds that are reminiscent of these backgrounds in section \ref{ssec:examples}.

\section{Non-Relativistic Geometry from Relativistic Geometry}

\noindent In order to obtain a matter-coupled non-relativistic susy QFT in three dimensions-- given by a Lagrangian, supersymmetry transformations and appropriate Killing spinor equations for the supersymmetry parameters--we apply a dimensional reduction along a lightlike isometry. As we saw above, any background that admits at least one solution of the Killing spinor equations \eqref{eq:Killspinoreqs}, has a null Killing vector $K^M$. We can describe the background geometry in coordinates that are adapted to this null Killing vector: $x^M = \{x^\mu, \mathsf{v}\}$, with $\mu = 0,1,2$, such that $K^M \partial_M = \partial_{\mathsf{v}}$. In these coordinates, the most general metric for which $K^M$ is a null Killing vector, can be described in terms of the following (inverse) Vielbein $E_M{}^A$ ($E^M{}_A$):
\begin{equation}
  \label{eq:vielbeinansatz}
  E_M{}^A = \bordermatrix{& a & - & + \cr
\mu & e_\mu{}^a &  \tau_\mu & - m_\mu \cr
\mathsf{v} & 0 & 0 & 1} \,, \qquad \qquad  E^M{}_A = \bordermatrix{& \mu & \mathsf{v} \cr
a & e^{\mu}{}_a & e^\mu{}_a m_\mu \cr
- & \tau^\mu &  \tau^\mu m_\mu \cr
+ & 0 & 1 }\,,
\end{equation}
where the flat indices $A = \{a,+,-\}$ refer to a null basis.
The $e^\mu{}_a$, $\tau^\mu$ that appear in the Ansatz for the inverse Vielbein $E^M{}_A$ are projective inverses of $e_\mu{}^a$, $\tau_\mu$, i.e.~they obey
\begin{align}
  \label{eq:invetau}
  & \tau^\mu \tau_\mu = 1 \,, \qquad \tau^\mu e_\mu{}^a = 0 \,, \qquad e^\mu{}_a \tau_\mu = 0 \,, \nonumber \\ & e^\mu{}_a e_\mu{}^b = \delta_a^b \,, \qquad e^\mu{}_a e^a{}_\nu = \delta^\mu_\nu - \tau^\mu \tau_\nu \,.
\end{align}
The $e_\mu{}^a$, $\tau_\mu$ and $m_\mu$ are independent of the $\mathsf{v}$-coordinate for $K^M$ to be a Killing vector. The form of the Vielbein \eqref{eq:vielbeinansatz} then corresponds to the Vielbein Ansatz that is used when performing a null reduction of the Einstein equations \cite{Duval:1984cj,Julia:1994bs}. The local space-time symmetries that are preserved in such a reduction, are given by the little group of the null Killing vector $K^M$. The Lie algebra of the little group of a null vector is given by the Bargmann algebra, the central extension of the algebra of Galilean space-time symmetries. One thus finds that local inertial frames in the lower-dimensional geometry are connected via Bargmann symmetries or in other words that the lower-dimensional geometry is Newton-Cartan. The quantities $e_\mu{}^a$ and $\tau_\mu$ then correspond to the spatial Vielbein and time-like Vielbein of a three-dimensional Newton-Cartan geometry. The field $m_\mu$ is a gauge field for the Bargmann U(1)-central charge symmetry with parameter $\beta$:
\begin{align} \label{eq:centralchargem}
    \delta m_\mu = \partial_\mu \beta\,.
\end{align}
From the null reduction viewpoint, this symmetry can be seen as stemming from infinitesimal diffeomorphisms in the $\mathsf{v}$-direction and its associated conserved charge is given by mass/particle number conservation. The field $m_\mu$ is a crucial ingredient in the Vielbein formulation of Newton-Cartan geometry \cite{Duval:1984cj}. Starting from the Vielbein Ansatz \eqref{eq:vielbeinansatz}, one can reduce other geometric quantities, such as the spin connection. Results for this are collected in appendix \ref{sec:NullReduction}.

The auxiliary scalar $U$ and vector field $V_M$ are also taken as independent of the $\mathsf{v}$-coordinate. We will rename
\begin{equation}
  u \equiv U \,,
\end{equation}
to distinguish the three-dimensional scalar $u$ from the four-dimensional one $U$.
It is convenient to redefine the reduced vector field $V_M$ as follows
\begin{align} \label{eq:defvmuv}
    v_\mu \equiv V_\mu + m_\mu\,V_{\mathsf{v}} = V_\mu + m_\mu v\,,\qquad v \equiv V_{\mathsf{v}}\,.
\end{align}
In this way, $v_\mu$ and $v$ are inert under the U(1)-central charge with parameter $\beta$, as are $e_\mu{}^a$, $\tau_\mu$ and $u$.

For future reference, we note that Galilean boosts with infinitesimal parameter $\lambda^a$ act as follows on the Newton-Cartan (inverse) Vielbeine and central charge gauge field:
\begin{alignat}{3} \label{eq:boosttrafosvielbs}
  \delta \tau_\mu &= 0 \,, \qquad & \delta e_\mu{}^a &= \lambda^a \tau_\mu \,, \qquad & \delta m_\mu &= - \lambda^a e_{\mu a} \,, \nonumber \\
  \delta \tau^\mu &= - \lambda^a e^\mu{}_a \,, \qquad & \delta e^\mu{}_a &= 0 \,.
\end{alignat}
The fields $u$ and $v$ are inert under boosts, while $v_\mu$ transforms as
\begin{align} \label{eq:boosttrafovmu}
  \delta v_\mu = -\lambda_a e_\mu{}^a v \,.
\end{align}

We will regularly turn three-dimensional lower indices $\mu$, $\nu$ on tensors into flat indices $0$, $a$, $(a=1,2)$, according to the rule
\begin{align}
  X_0 = \tau^\mu X_\mu \,, \qquad X_a = e^\mu{}_a X_\mu \,.
\end{align}
The $a$ index can be freely raised and lowered using a Kronecker delta. We will take $X^0 = - X_0$.
\subsection{Scherk-Schwarz Null Reduction}
The easiest way to perform the null reduction for the matter multiplet consists of using the Ansatz \eqref{eq:vielbeinansatz} and assuming that the (anti-)chiral multiplet fields are $\mathsf{v}$-independent. It is easy to see that this leads to a Lagrangian without time derivatives for the physical scalars, such that these scalars obey Poisson-type equations of motion. We will not discuss this case further; instead we will focus on a reduction that leads to dynamical fields that obey Schr\"odinger-type equations of motion. This can be achieved by performing a {\it twisted} or {\it Scherk-Schwarz} reduction \cite{Scherk:1979hw}. Such a reduction can be applied whenever the higher-dimensional theory has a global symmetry. One can then propose an Ansatz in which the higher-dimensional fields are expressed as symmetry transformations of the lower-dimensional fields, where the symmetry transformations depend on the internal coordinates. Invariance of the higher-dimensional theory under the symmetry then guarantees that this is a consistent reduction Ansatz, i.e.~that the dependence on the internal coordinates drops out when plugging the Ansatz into the higher-dimensional quantities.

In order to perform the Scherk-Schwarz reduction, we will assume that the Lagrangian \eqref{eq:Lrigid} exhibits the following global U$(1)$-symmetry, with parameter $\alpha$:
\begin{alignat}{3}
  \label{eq:globalU1}
  \delta Z &= \rmi\, \alpha\, Z \,, \qquad \qquad & \delta \chi_L &= \rmi\, \alpha\, \chi_L \,, \qquad \qquad & \delta H &= \rmi\, \alpha\, H \,.
\end{alignat}
This happens when the superpotential $W$ is zero and we will thus take $W=0$ from now on.\footnote{Note that choosing $W=0$ excludes interesting interaction terms. This restriction can however be lifted by e.g. introducing extra chiral multiplets such that a U$(1)$-invariant superpotential can be engineered. This was for instance done in \cite{Auzzi:2019kdd} to obtain an interacting non-relativistic Wess-Zumino model in flat space via Scherk-Schwarz null reduction of a relativistic one.} We can then use this U$(1)$-symmetry to perform the twisted null reduction. We thus propose the following Ansatz for the bosonic chiral multiplet fields in terms of three-dimensional scalars $z(x^\mu)$, $h(x^\mu)$:
\begin{alignat}{2}
  \label{eq:twistansatzbos}
  Z(x^\mu,\mathsf{v}) &= \rme^{-\rmi \, m \, \mathsf{v}} z(x^\mu)\,, \qquad \qquad &  H(x^\mu,\mathsf{v}) &= \rme^{-\rmi \, m \, \mathsf{v}} h(x^\mu) \,.
\end{alignat}
In order to give the reduction Ansatz for the fermion $\chi_L$, $\chi_R$, we adopt a decomposition of the four-dimensional Clifford algebra in terms of the three-dimensional one, discussed in appendix \ref{sec:NullReduction}. We then propose the following reduction Ansatz 
\begin{subequations}
\begin{align}
    \chi_L(x^\mu,\mathsf{v}) &= \rme^{-\rmi \,m\,\mathsf{v}}\Bigl(\pi\psi_+(x^\mu)\otimes\varphi_- + \bar{\pi}\psi_-(x^\mu)\otimes\varphi_+\Bigr),\\[.2truecm]
    \chi_R(x^\mu,\mathsf{v}) &=\rme^{+\rmi\,m\,\mathsf{v}}\Big( \bar{\pi}\psi_+(x^\mu)\otimes\varphi_- + \pi\psi_-(x^\mu)\otimes\varphi_+\Big).
\end{align}
\label{eq:twistansatzferm1}
\end{subequations}
Here $\psi_\pm$ are three-dimensional Majorana spinors (obeying the three-dimensional Majorana condition $\psi^*_\pm = \rmi\mathcal{C}_3\gamma^0\psi_\pm$) and $\varphi_+ = (1,\,0)^T$ and $\varphi_- = (0,\,1)^T$ obey
\begin{equation}\sigma_\pm\varphi_\pm = 0\hskip .5truecm \textrm{and}\hskip .5truecm  \sigma_\pm\varphi_\mp = \sqrt{2}\,\varphi_\pm\,,
\end{equation}
with the matrices $\sigma_\pm$ defined in \eqref{eq:defsigmapm}.
In \eqref{eq:twistansatzferm1}, we have used the three-dimensional operators $\pi$, $\bar{\pi}$, that are defined as
\begin{equation}
    \pi = \frac12\big(\mathds{1}_2 - \rmi\gamma_0\big)\qquad \text{and}\qquad \bar{\pi} = \frac12\big(\mathds{1}_2 + \rmi\gamma_0\big)\,. \label{eq:quasichiralproj}
\end{equation}
Since $\big(\rmi\gamma_0\big)^2=\mathds{1}_2$ these operators are projectors, that satisfy  $\rmi\gamma_0\pi = -\pi$ and $\rmi\gamma_0\bar{\pi} = \bar{\pi}$. With a slight abuse of terminology, we will refer to $\{\pi\psi_+,\bar{\pi}\psi_-\}$ as (pseudo-)left-handed fermions and to $\{\bar{\pi}\psi_+,\pi\psi_-\}$ as (pseudo-)right-handed fermions, alluding to their four-dimensional origin. Note that these pseudo-right-handed and pseudo-left-handed fermions are no longer Majorana, but are instead complex one-component spinors.

As mentioned above, when performing the null reduction, one finds that the lower-dimensional local symmetries span the Bargmann algebra, that includes local spatial rotations, local Galilean boosts and a local U$(1)$-central charge transformation, for which $m_\mu$ is a gauge field (see \eqref{eq:centralchargem}). This local U(1)-central charge that is associated to mass/particle number conservation acts on the three-dimensional fields $z(x)$, $\psi_\pm(x)$, $h(x)$ as follows:
\begin{align}\label{eq:partnumber}
  \delta_{\mathrm{U}(1)} z(x) &= \rmi\,m\,\beta \,z(x)\,,\qquad \qquad \delta_{\mathrm{U}(1)} \psi_\pm(x) = \pm m\,\beta \,\gamma_0 \,\psi_\pm(x)\,, \nonumber \\
  \delta_{\mathrm{U}(1)}h(x) &= \rmi\,m\,\beta\,h(x) \,.
  \end{align}
The reduction of the four-dimensional (anti-)chiral multiplet $\{Z,\chi_L, H\}$ ($\{\bar{Z},\chi_R, \bar{H}\}$) then leads to a three-dimensional pseudo-(anti-)chiral multiplet $\{z,\pi \psi_+, \bar{\pi}\psi_-, h\}$ ($\{\bar{z},\linebreak[1]\bar{\pi} \psi_+,\linebreak[1] \pi\psi_-,\linebreak[1] \bar{h}\}$).
In the following, we will use covariant derivatives $\bar \nabla_\mu$ in three dimensions, that are covariantized with respect to local rotations, Galilean boosts and the U(1)-transformations \eqref{eq:partnumber}. When acting on the physical fields of the three-dimensional pseudo-chiral multiplet, these derivatives are defined as follows:
\begin{subequations}
\begin{align} 
    \bar{\nabla}_\mu z &= \partial_\mu z - \rmi\,m\,m_\mu z\,,\\
    \bar{\nabla}_\mu \pi\psi_+ &= \partial_\mu\pi\psi_+ - \rmi\,m\,m_\mu\pi\psi_+ + \frac14\omega_\mu{}^{ab}\gamma_{ab}\pi\psi_+\,,\\
    \bar{\nabla}_\mu \bar{\pi}\psi_- &= \partial_\mu\bar{\pi}\psi_- - \rmi\,m\,m_\mu\bar{\pi}\psi_- + \frac14\omega_\mu{}^{ab}\gamma_{ab}\bar{\pi}\psi_- - \frac{\rmi\sqrt2}{2}\omega_\mu{}^a\gamma_a\pi\psi_+\,,
\end{align}
\label{eq:3covdev}\end{subequations}
where the spin connections $\omega_\mu{}^{ab}$, $\omega_\mu{}^a$ for local spatial rotations and Galilean boosts depend on $\tau_\mu$, $e_\mu{}^a$, $m_\mu$. Their explicit expressions can be found in eq.\,\eqref{eq:defNCconnstaumunu}. Similar expressions can be obtained by complex conjugation for the fields of the pseudo-anti-chiral multiplet.

\subsection{Multiplets and Lagrangian}
In this section, we will construct an explicit example of a non-relativistic susy QFT, coupled to an arbitrary curved Newton-Cartan background. The resulting theory is a supersymmetric extension of a field theory for a scalar, that obeys a curved space Schr\"odinger equation. The fermions obey a Levy-Leblond equation \cite{LevyLeblond:1967zz}--which can be seen as the square root of the Schr\"odinger equation, similar to how the Dirac equation can be viewed as the square root of the Klein-Gordon equation. It has been proposed recently \cite{Auzzi:2019kdd}, that an interacting version of this theory in flat space is one-loop exact. We opted to consider only one pseudo-chiral multiplet for pedagogical reasons. The generalization to an arbitrary number of pseudo-chiral multiplets, with arbitrary K\"ahler potentials and potentially non-zero superpotentials, is straightforward.

Applying the above Ans\"atze to eq. \eqref{eq:Lrigid}, we find the following Lagrangian 
\begin{align}\label{eq:3dact}
    \lr\det(e^a,\tau)\rr^{-1}\mathcal{L} &= \lr\frac{1}{12}\, \epsilon^{ab} R_{ab}\lr J\rr - \frac13\,\tau^{a0}\tau^{a0} + \frac19\,\bar{u}u - \frac19\,h^{\mu\nu} v_\mu v_\nu + \frac29\,v\tau^\mu v_\mu\rr z \bar{z}\notag\\
    &\,-h^{\mu\nu}\bar{\nabla}_\mu z\,\bar{\nabla}_\nu \bar{z} + \rmi\,m\,\tau^\mu\lr \bar{z}\bar{\nabla}_\mu z - z\bar{\nabla}_\mu \bar{z}\rr + h\bar{h}\notag\\
    &\,-e^{\mu a}\bar{\psi}_-\gamma_a\bar{\nabla}_\mu\psi_+ - e^{\mu a}\bar{\psi}_+\gamma_a\bar{\nabla}_\mu\psi_- + \sqrt2\,\tau^\mu\bar{\psi}_+\gamma_0\lr\bar{\nabla}_\mu-\frac16\,v_\mu\gamma_0\rr\psi_+\notag\\
    &\,-m\sqrt2\,\bar{\psi}_-\psi_- + \frac{\sqrt2}{8}\lr\tau^{ab}\epsilon_{ab} - \frac43\,v\rr\bar{\psi}_-\psi_- +\frac12\lr \tau^{a0} + \frac23\epsilon^{ab}v^b\rr\bar{\psi}_+\gamma_a\psi_-\notag\\
    &\,+\frac13\left(\bar{u}\,\bar{z}\,h + u\,z\,\bar{h}\right) +\frac{\rmi}{3}\lr h^{\mu\nu}v_\nu- v\tau^\mu\rr\lr\bar{z}\bar{\nabla}_\mu z - z\bar{\nabla}_\mu\bar{z}\rr - \frac23\,v_0\,m\,z \bar{z} \,.
\end{align}
Here, the notation $\det(e^a,\tau)$ refers to the determinant of a $(3\times 3)$-matrix, obtained by putting $e_\mu{}^a$ and $\tau_\mu$ in its columns. We have also defined the so-called spatial metric of Newton-Cartan geometry $h^{\mu\nu}$ as $h^{\mu\nu} = e^{\mu a} e^{\nu}{}_a$. The notation $\tau_{ab}$, resp. $\tau_{0a}$ refers to the spatial, resp. time-like parts of the curl of $\tau_\mu$
\begin{equation}
  \tau_{ab} = 2 e^\mu{}_a e^\nu{}_b \partial_{[\mu} \tau_{\nu]} \,, \qquad \qquad \tau_{0a} = 2 \tau^\mu e^\nu{}_a \partial_{[\mu} \tau_{\nu]} \,.
\end{equation}
The curvature of spatial rotations $R_{\mu\nu}(J)$ that appears in the first term is defined in eq. \eqref{eq:Jcurv}.

The reduction of the four-dimensional supersymmetry transformation rules leads to the following supersymmetry transformation rules for the pseudo-(anti-)chiral multiplet
\begin{align}\label{eq:3dtrans}
    \delta z &= \bar{\epsilon}_+\bar{\pi}\psi_- + \bar{\epsilon}_-\pi\psi_+\,,\notag\\
    \delta \pi\psi_+ &= \frac12\,e^{\mu a}\gamma_a\bar{\pi}\epsilon_+\bar{\nabla}_\mu z + \frac12\,h\pi\epsilon_+ + \frac{\rmi\,m}{\sqrt2}\,z\,\gamma_0\pi\epsilon_-\,,\notag\\
    \delta \bar{\pi}\psi_- &= \frac12\,e^{\mu a}\gamma_a\pi\epsilon_-\bar{\nabla}_\mu z + \frac12\,h\bar{\pi}\epsilon_- - \frac{1}{\sqrt2}\,\tau^\mu\,\gamma_0\bar{\pi}\epsilon_+\bar{\nabla}_\mu z\,,\\
    \delta h &= e^{\mu a}\bar{\epsilon}_-\gamma_a\lr\bar{\nabla}_\mu - \frac16 \,v_\mu\gamma_0\rr\pi\psi_+ + e^{\mu a}\bar{\epsilon}_+\gamma_a\lr\bar{\nabla}_\mu + \frac16 \,v_\mu\gamma_0\rr\bar{\pi}\psi_- \notag\\
    &\,-\sqrt2\,\tau^\mu\,\bar{\epsilon}_+\gamma_0\lr\bar{\nabla}_\mu - \frac16 \,v_\mu\gamma_0\rr\pi\psi_+ - \frac{u}{3}\lr\bar{\epsilon}_+\bar{\pi}\psi_- + \bar{\epsilon}_-\pi\psi_+\rr\notag\\
    &\,+m\sqrt2\,\bar{\epsilon}_-\bar{\pi}\psi_- -\frac{\sqrt2}{8}\lr\tau^{ab}\epsilon_{ab} - \frac43\,v\rr\bar{\epsilon}_-\bar{\pi}\psi_-\notag\\
    &\,- \frac14\,\tau^{a0}\lr \bar{\epsilon}_-\gamma_a\pi\psi_+ + 3\,\bar{\epsilon}_+\gamma_a\bar{\pi}\psi_-\rr\,.\notag
\end{align}
Here, it is understood that $(\epsilon_+,\epsilon_-)$ solves the Killing spinor equations, given in the next section. The Lagrangian \eqref{eq:3dact} is then invariant under \eqref{eq:3dtrans} up to total derivatives, when using the modified rule for partial integration \eqref{eq:partialIntegration}.

\subsection{Killing Spinor Equations for Non-Relativistic Supersymmetry} \label{ssec:KSeqs}

In order to establish the coupling to concrete backgrounds, we consider the Killing spinor equations obtained from the null reduction of eqs.\,\eqref{eq:Killspinoreqs}. It is worth mentioning, that the supercharges $(\epsilon_+,\epsilon_-)$ have charge zero under the U(1)-central charge transformation, hence the reduction is to be understood as an ordinary null reduction. This leads to four independent equations, two of which are purely algebraic:
\begin{subequations}
    \begin{align}
        &4\, v\, \gamma_0\epsilon_++\tau^{ab}\gamma_{ab}\epsilon_+ =0\,, \label{eq:OMalg1}\\
        & v \gamma_0\epsilon_- - \frac{3}{4} \tau^{ab} \gamma_{ab} \epsilon_-  - \frac{3\sqrt{2}}{2} \tau^{a0} \gamma_{a0} \epsilon_+ + \sqrt2\, v_a \gamma^a \epsilon_+ \notag \\
        & \quad \ \  - \sqrt2\, \mathrm{Re}(u) \gamma_0 \epsilon_+ + \sqrt2 \, \mathrm{Im}(u) \epsilon_+=0\,, \label{eq:OMalg2}
    \end{align}
  \end{subequations}
and two of which are differential equations
\begin{subequations}
\begin{align}
\bar{\nabla}_\mu \epsilon_+ &=- \frac14 {\tau}_\mu{}^0 \epsilon_+ - \frac{\sqrt{2}}{4} {\tau}_\mu{}^a \gamma_{a0} \epsilon_-   - \frac12 v_\mu \gamma_0 \epsilon_+ +\frac16 e_\mu{}^a v_b \gamma_a \gamma^b \gamma_0 \epsilon_+ + \frac13 \tau_\mu v_0 \gamma_0 \epsilon_+ \nonumber\\
&\quad\, - \frac{\sqrt2}{6} \tau_\mu v_a \gamma^a \epsilon_- - \frac{\sqrt2}{6} v\, e_\mu{}^a  \gamma_a \epsilon_-  - \frac16 \mathrm{Re}(u)\,e_\mu{}^a \gamma_a \epsilon_+ - \frac{\sqrt{2}}{6}  \mathrm{Re}(u)\,\tau_\mu\gamma_0 \epsilon_-\nonumber \\
& \quad\, - \frac16\mathrm{Im}(u) e_\mu{}^a \gamma_{a0} \epsilon_+  - \frac{ \sqrt{2}}{6}  \mathrm{Im}(u) \tau_\mu\epsilon_-  \,, \label{eq:ucKSE1} \\
\bar{\nabla}_\mu \epsilon_- &=  +\frac14 {\tau}_\mu{}^0 \epsilon_- +\frac12 v_\mu \gamma_0 \epsilon_- - \frac16 e_\mu{}^a v_b \gamma_a \gamma^b \gamma_0 \epsilon_-  + \frac{\sqrt{2}}{6} e_\mu{}^a v_0 \gamma_a \epsilon_+ \nonumber\\
&\quad\, - \frac16 \mathrm{Re}(u)\,e_\mu{}^a \gamma_a \epsilon_- + \frac16\mathrm{Im}(u)\,e_\mu{}^a\gamma_{a0} \epsilon_-\,, \label{eq:ucKSE2}
\end{align}
\end{subequations}
where the covariant derivatives on $\epsilon_\pm$ are explicitly given by
\begin{align}
  \label{eq:covderepspm}
  \bar{\nabla}_\mu \epsilon_+ &=  \partial_\mu \epsilon_+ + \frac14 \omega_\mu{}^{ab} \gamma_{ab} \epsilon_+ \,, \nonumber \\
  \bar{\nabla}_\mu \epsilon_- &=  \partial_\mu \epsilon_- + \frac14 \omega_\mu{}^{ab} \gamma_{ab} \epsilon_- - \frac{\sqrt{2}}{2} \omega_\mu{}^a \gamma_{a0} \epsilon_+ \,.
\end{align}
This set of two algebraic and two differential Killing spinor equations is invariant under local Galilean boosts, under which the background fields and spin connections transform as in eqs. \eqref{eq:boosttrafosvielbs}, \eqref{eq:boosttrafovmu}, \eqref{eq:boosttrafoconns}, and under which $\epsilon_\pm$ transform as
\begin{align} \label{eq:boostepsilons}
  \delta \epsilon_+ &= 0 \,, \qquad \qquad \delta \epsilon_- = -\frac{\sqrt{2}}{2} \lambda^a \gamma_{a0} \epsilon_+ \,.
\end{align}
The boost invariance of this set of equations is slightly non-trivial. One can show that under boosts the second algebraic Killing spinor equation \eqref{eq:OMalg2} transforms to the first algebraic one \eqref{eq:OMalg1}. The first differential Killing spinor equation \eqref{eq:ucKSE1} transforms to the first algebraic one \eqref{eq:OMalg1}. The second differential Killing spinor equation \eqref{eq:ucKSE2} transforms to a combination of the first differential one \eqref{eq:ucKSE1} and the second algebraic one \eqref{eq:OMalg2}. While the inclusion of algebraic equations as part of the non-relativistic Killing spinor equations might seem strange at first, one sees that they are necessary to obtain a set of equations that is invariant under these local Galilean boosts.

It is worth comparing this null reduction of the Killing spinor equations with a reduction of the four-dimensional Killing spinor equations along a spatial isometry \cite{Dumitrescu:2012ha,Closset:2012ru}. Also in the latter case, dimensional reduction leads to a set of differential and a set of algebraic Killing spinor equations. In that case however, all three-dimensional Killing spinor equations are Lorentz-covariant on their own and the algebraic Killing spinor equations decouple from the differential ones in the sense that one only needs to consider the latter when determining which backgrounds admit Killing spinors. The underlying reason for this is that after spatial reduction, the Old Minimal supergravity multiplet gives a fully reducible representation of the three-dimensional super-Poincar\'e algebra and splits into the three-dimensional supergravity multiplet and an extra matter multiplet that can be truncated. The differential Killing spinor equations then correspond to the supersymmetry transformations of the gravitini of the off-shell supergravity multiplet. The algebraic ones on the other hand correspond to the supersymmetry transformation rules of the fermions of the matter multiplet and hence do not need to be considered when looking for suitable Killing spinors. 
	
	This conclusion changes when considering a reduction along a lightlike direction. In that case the four-dimensional supergravity multiplet reduces to an indecomposable reducible representation of the three-dimensional super-Bargmann algebra and no longer splits nicely into a three-dimensional supergravity multiplet and an extra matter multiplet. Fields that would sit in a matter multiplet upon spatial reduction no longer do so upon null reduction, as they can be linked by Galilean boosts to other supergravity multiplet fields. It is for this reason that the boost transformation of the differential Killing spinor equations leads to the algebraic ones and that we keep the algebraic equations in order to perform the most general analysis of which non-relativistic backgrounds preserve supersymmetry.

\section{Solutions}

In the above section, we found a set of algebraic and differential equations that the non-relativistic Killing spinors obey. One is able to define supersymmetry on a given background, whenever these Killing spinor equations in this background admit non-trivial, nowhere vanishing,\footnote{In practice, the requirement that the solution is nowhere vanishing is often automatic if the solution is non-trivial; see the discussion around equation \eqref{eq:simplepdeform}.} solutions. Indeed, in that case one can use these solutions as a basis for the supersymmetry parameters appearing in \eqref{eq:3dtrans}. Since some of the Killing spinor equations are partial differential equations, they do not exhibit non-trivial solutions for all possible backgrounds. The allowed backgrounds for instance have to comply with the integrability conditions for the differential Killing spinor equations and there might also be topological obstructions to the existence of suitable Killing spinors. In this section, we will investigate the constraints that backgrounds have to obey, such that non-trivial non-relativistic Killing spinors can be found. We will also give some examples of such backgrounds.

In identifying the allowed backgrounds, we will adapt techniques that are similar to the ones used in the relativistic four-dimensional case \cite{Festuccia:2011ws,Dumitrescu:2012ha,Dumitrescu:2012at,Liu:2012bi,Klare:2012gn,Cassani:2012ri} to the situation at hand, e.g. taking into account that we now also have algebraic Killing spinor equations. This will lead to conditions on the backgrounds that are necessary and sufficient for the Killing spinor equations to have non-trivial solutions. Necessary conditions can also be obtained from studying the integrability conditions for the Killings spinor equations. These integrability conditions are often useful for practical purposes, e.g. when analyzing particular backgrounds. For this reason, we have discussed them in detail in appendix \ref{sec:Integrability}. The analysis of the integrability conditions offers an alternative viewpoint to the results of subsections \ref{ssec:zetapluszero} and \ref{ssec:zetaplusneq0} and on top of that it provides some additional explicit formulas that are useful in the examples of subsection \ref{ssec:examples}.  

As in the four-dimensional case discussed in section \ref{sec:susy4d}, we will be interested in commuting solutions $(\zeta_+,\zeta_-)$ of \eqref{eq:OMalg1}--\eqref{eq:ucKSE2}. Given a basis of nowhere vanishing solutions $\left\{\left(\zeta_+^{(i)}, \zeta_-^{(i)}\right) | i = 1,\cdots,n\right\}$ (where $1 \leq n \leq 4$), the rigid supersymmetry parameters $(\epsilon_+ = \theta \zeta_+, \epsilon_-=\theta \zeta_-)$ can then be constructed by multiplying these basis solutions with arbitrary constant Grassmann parameters $\theta$. In order to find such a basis of commuting solutions $(\zeta^{(i)}_+, \zeta^{(i)}_-)$, let us first note that the first algebraic Killing spinor equation \eqref{eq:OMalg1} evaluated on a generic solution $(\zeta_+, \zeta_-)$, is equivalent to
\begin{align}
  \left(4\, v + \tau^{ab} \epsilon_{ab} \right) \zeta_+ = 0 \,.
\end{align}
This equation suggests that the search for solutions can be subdivided into a case in which one looks for solutions where $\zeta_+$ is identically zero and a case where $\zeta_+$ is not identically zero (but $4v + \tau^{ab} \epsilon_{ab}$ is). We will now discuss both cases in turn.

\subsection[The case $\zeta_+ = 0$]{The case \boldmath$\zeta_+ = 0$} \label{ssec:zetapluszero}

In this case, we are looking for Killing spinors of the form $(0,\zeta_-)$, where $\zeta_-$ solves the following remaining Killing spinor equations \eqref{eq:OMalg2}, \eqref{eq:ucKSE1}, \eqref{eq:ucKSE2}
\begin{subequations}
\begin{align}
  &  \left(\frac43 v  - \tau^{ab} \epsilon_{ab}\right) \gamma_0 \zeta_- = 0\,, \label{eq:algepsp01} \\
  & \left(\frac32 \tau_\mu{}^a \gamma_{a0}  + e_\mu{}^a v \gamma_a  + \tau_\mu v^a \gamma_a  + \mathrm{Re}(u) \tau_\mu \gamma_0 + \mathrm{Im}(u) \tau_\mu\right) \zeta_- = 0 \,, \label{eq:algepsp02} \\
  & D_\mu \zeta_- = \Big(\frac14 {\tau}_\mu{}^0  +\frac12 v_\mu \gamma_0  - \frac16 e_\mu{}^a v_b \gamma_a \gamma^b \gamma_0   - \frac16 \mathrm{Re}(u)\,e_\mu{}^a \gamma_a  \nonumber \\ & \qquad \qquad \ \, + \frac16\mathrm{Im}(u)\,e_\mu{}^a\gamma_{a0} \Big) \zeta_-\,, \label{eq:diffepsp0}
\end{align}
\end{subequations}
with
\begin{align}
  D_\mu \zeta_- = \partial_\mu \zeta_- + \frac14 \omega_\mu{}^{ab} \gamma_{ab} \zeta_- \,.
\end{align}
Note in particular that the first differential Killing spinor equation \eqref{eq:ucKSE1} has turned into an algebraic equation.

Before discussing the constraints on the background geometry and auxiliary fields that follow from requiring the existence of non-trivial solutions of eqs. \eqref{eq:algepsp01}--\eqref{eq:diffepsp0}, let us first note that one can reasonably assume that any non-trivial solution $\zeta_-$ of these equations is nowhere vanishing. Indeed, the differential equation \eqref{eq:diffepsp0} is of the form
\begin{align} \label{eq:simplepdeform}
  \partial_\mu \zeta_- = B_\mu \zeta_- \,,
\end{align}
where $B_\mu$ is a Clifford algebra valued operator that involves geometric quantities and auxiliary fields. Suppose then that there exists a point $p$, where $\zeta_-$ is zero ($\zeta_-|_p = 0$). Equation \eqref{eq:simplepdeform} then implies that also $\partial_\mu \zeta_- |_p = 0$. Similarly, by taking successive partial derivatives of \eqref{eq:simplepdeform}, one can iteratively infer that all partial derivatives of $\zeta_-$ vanish at $p$. If $\zeta_-|_p = 0$, we thus find that the Taylor series of $\zeta_-$ around $p$ vanishes identically and consequently, assuming reasonable analyticity properties for $\zeta_-$, that $\zeta_-$ is given by the trivial zero solution. Non-trivial solutions for $\zeta_-$ can therefore be assumed to be nowhere vanishing and we will do so in the following.

With this in mind, we can discuss the conditions under which the equations \eqref{eq:algepsp01}--\eqref{eq:diffepsp0} admit non-trivial solutions. We can phrase these conditions in the form of the following theorem, which is the basic result of this subsection.
\begin{theorem}\label{thm1}
  The equations \eqref{eq:algepsp01}--\eqref{eq:diffepsp0} have one non-trivial globally well-defined solution for $\zeta_-$ if and only if there exists a globally well-defined unit vector $X^-_a$ such that the following conditions hold:{\footnote{Note that $v^{b}=e^{\mu b}v_{\mu}$, which appears on the right-hand side of condition \eqref{eq:soltau0a1}, is fully determined by condition \eqref{eq:solv1}. Here we do not substitute its explicit expression for brevity.}}
  \begin{subequations}
  \begin{align}
    \epsilon^{ab} \tau_{ab} &= \frac43 v \label{eq:vistauab} \,, \\
    \tau_{0a} &= \frac23 \left(-\epsilon_{ab} v^b + \mathrm{Re}(u) X^-_a - \mathrm{Im}(u) Y^-_a \right) \label{eq:soltau0a1} \,, \\
    v_\mu &= \tau_\mu Y^{-a} D_0 X^-_a + \frac12 e_\mu{}^a \left(3 Y^{-b} D_a X^-_b + \mathrm{Re}(u) Y^-_a + \mathrm{Im}(u) X^-_a \right) \label{eq:solv1} \,,
  \end{align}
  \end{subequations}
  where $Y^-_a = \epsilon_{ab} X^{-b}$ and $D_\mu X^-_a = \partial_\mu X^-_a + \omega_{\mu a}{}^b X^-_b$. There are two independent globally well-defined solutions for $\zeta_-$ if and only if there exists a globally well-defined unit vector $X^-_a$ such that the conditions \eqref{eq:vistauab}--\eqref{eq:solv1} hold with $u=0$.
\end{theorem}
\begin{proof}
In order to prove this theorem, let us first assume that one globally well-defined, nowhere vanishing, solution $\zeta_-^{(1)}$ of eqs. \eqref{eq:algepsp01}--\eqref{eq:diffepsp0} exists and let us show that this implies the conditions \eqref{eq:vistauab}--\eqref{eq:solv1}. Equation \eqref{eq:algepsp01}, evaluated on this solution, is equivalent to
\begin{align}
  \left(\frac43 v - \tau^{ab} \epsilon_{ab} \right) \zeta_-^{(1)} = 0 \,.
\end{align}
Since $\zeta_-^{(1)}$ is assumed to be nowhere vanishing, we thus see that \eqref{eq:vistauab} has to hold. We can then use this condition in equation \eqref{eq:algepsp02}, evaluated on $\zeta_-^{(1)}$. Doing this, one finds (after multiplication with $\tau^\mu$) that 
\begin{align} \label{eq:auxalgeq}
  \frac32 \epsilon^{ab} \tau_{0a} \gamma_b \zeta_-^{(1)} +  v^a \gamma_a \zeta_-^{(1)} + \mathrm{Re}(u) \gamma_0 \zeta_-^{(1)} + \mathrm{Im}(u) \zeta_-^{(1)} = 0 \,.
\end{align}
In order to proceed, we note that one can use the nowhere vanishing and globally well-defined solution $\zeta_-^{(1)}$ to construct the following bilinears
\begin{align} \label{eq:defNXY}
  N^- = \rmi \bar{\zeta}_-^{(1)} \gamma_0 \zeta_-^{(1)} \,, \qquad X^-_a = \frac{1}{N^-} \rmi \bar{\zeta}_-^{(1)} \gamma_{0a} \zeta_-^{(1)} \,, \qquad Y^-_a = \frac{1}{N^-} \rmi \bar{\zeta}_-^{(1)} \gamma_a \zeta_-^{(1)} \,.
\end{align}
Since $N^-$ is given by $-\left(\zeta_-^{(1)}\right)^\dag \zeta_-^{(1)}$, it is nowhere vanishing because $\zeta_-^{(1)}$ is. The vectors $X^-_a$ and $Y^-_a$ are thus well-defined. They are not independent; rather they are related by
\begin{align} \label{eq:XperpY}
  X^-_a = - \epsilon_{ab} Y^-_b \,.
\end{align}
Fierz identities moreover imply that $X^-_a$ and $Y^-_a$ are unit vectors (and thus nowhere vanishing)
\begin{align} \label{eq:XYunit}
  X^{-a} X^-_a = 1 \,, \qquad \qquad \qquad Y^{-a} Y^-_a = 1 \,,
\end{align}
and that they obey
\begin{align} \label{eq:propXY}
  X^{-a} \gamma_a \zeta_-^{(1)} = \zeta_-^{(1)} \,, \qquad \qquad Y^{-a} \gamma_a \zeta_-^{(1)} = \gamma_0 \zeta_-^{(1)} \,.
\end{align}
These properties can then be used to rewrite \eqref{eq:auxalgeq} as
\begin{align} \label{eq:Agammaeps}
  A^a \gamma_a \zeta_-^{(1)} = 0 \,, \qquad \quad \mathrm{where} \qquad \quad A^a = \frac32 \tau_{0b} \epsilon^{ba} + v^a + \mathrm{Re}(u) Y^{-a} + \mathrm{Im}(u) X^{-a} \,.
\end{align}
Since $\zeta^{(1)}_-$ is non-trivial, this equation expresses that the matrix $A^a \gamma_a$ is singular and thus that its determinant is zero. Since
\begin{align}
  \det(A^a \gamma_a)^2 = (A^a A_a)^2 \,,
\end{align}
we thus see that \eqref{eq:Agammaeps} implies that $A_a = 0$ or in other words that \eqref{eq:soltau0a1} holds. We can then use \eqref{eq:vistauab} and \eqref{eq:soltau0a1}, along with \eqref{eq:propXY} in the differential condition \eqref{eq:diffepsp0} on $\zeta_-^{(1)}$, leading to the following equation:
\begin{align}
  &  D_\mu \zeta_-^{(1)} = C^-_\mu \gamma_0 \zeta_-^{(1)} \,, \nonumber \\
  & \mathrm{where} \qquad C^-_\mu = \frac12 \tau_\mu v_0 + \frac13 e_\mu{}^a v_a - \frac16 \mathrm{Re}(u) e_\mu{}^a Y^-_a - \frac16 \mathrm{Im}(u) e_\mu{}^a X^-_a \,.
\end{align}
Using this equation and the definitions \eqref{eq:defNXY}, one can show that
\begin{align}
  \partial_\mu N^- = 0 \,, \qquad \mathrm{and} \qquad D_\mu X^-_a = 2 C^-_\mu Y^-_a \,.
\end{align}
The latter equation implies that
\begin{align} \label{eq:bigCmin}
  C_\mu^- = \frac12 Y^{-a} D_\mu X^-_a \,,
\end{align}
which can be rewritten as the third condition \eqref{eq:solv1}.

Similar steps can be taken to show that \eqref{eq:vistauab}--\eqref{eq:solv1} hold with $u=0$, when there exists a second solution $\zeta_-^{(2)}$ of eqs. \eqref{eq:algepsp01}--\eqref{eq:diffepsp0}, that is linearly independent from $\zeta_-^{(1)}$. Note that we can always write $\zeta_-^{(2)}$ as a linear combination of the eigenvectors $\{\zeta_-^{(1)}, \gamma_0 \zeta_-^{(1)}\}$ (with eigenvalues $+1$ and $-1$ resp.) of $X^{-a} \gamma_a$: $\zeta_-^{(2)} = a \zeta_-^{(1)} + b \gamma_0 \zeta_-^{(1)}$, with $a$, $b \in \mathbb{R}$ and $b\neq 0$. Linearity of the Killing spinor equations then implies that we can take
\begin{equation}
  \zeta_-^{(2)} = \gamma_0 \zeta_-^{(1)} 
\end{equation}
without loss of generality and we will adopt this choice in the following. Evaluating equation \eqref{eq:algepsp01} on this second solution then again leads to \eqref{eq:vistauab}. Considering equation \eqref{eq:algepsp02}, evaluated on $\zeta_-^{(2)}$, and performing manipulations similar to those that led to \eqref{eq:soltau0a1}, now implies that 
\begin{equation} 
\tau_{0a} = \frac23 \left(-\epsilon_{ab} v^b - \mathrm{Re}(u) X^-_a + \mathrm{Im}(u) Y^-_a \right)
\end{equation}
should hold along with \eqref{eq:soltau0a1}. This is only possible when $u=0$ and we thus find that \eqref{eq:soltau0a1} holds with $u=0$. Using \eqref{eq:vistauab}, \eqref{eq:soltau0a1} and $u=0$ in the differential condition \eqref{eq:diffepsp0}, evaluated on $\zeta_-^{(1)}$ (or, giving equivalent results, on $\zeta_-^{(2)} = \gamma_0 \zeta_-^{(1)}$), then leads to
\begin{align}
  &  D_\mu \zeta_-^{(1)} = c^-_\mu \gamma_0 \zeta_-^{(1)} \,, \nonumber \\
  & \mathrm{where} \qquad c^-_\mu = \frac12 \tau_\mu v_0 + \frac13 e_\mu{}^a v_a  \,.
\end{align}
The same reasoning that led to \eqref{eq:bigCmin} can then be used to show that
\begin{equation}
  c_\mu^- = \frac12 Y^{-a} D_\mu X^-_a \,,
\end{equation}
which is equivalent to \eqref{eq:solv1} with $u=0$.
This completes the proof that the existence of a non-trivial globally well-defined solution of the form $(0,\zeta_-)$ of the Killing spinor equations implies the existence of a globally well-defined unit vector $X^-_a$ such that eqs. \eqref{eq:vistauab}, \eqref{eq:soltau0a1} and \eqref{eq:solv1} hold, with $u=0$ in case two such solutions exist.

Let us now prove that the reverse statement also holds and assume that \eqref{eq:vistauab}--\eqref{eq:solv1} hold for a globally well-defined unit vector $X^-_a$. Note first that eq. \eqref{eq:algepsp01} is identically satisfied for any $\zeta_-$ when \eqref{eq:vistauab} holds. Using \eqref{eq:vistauab} and \eqref{eq:soltau0a1} in eq. \eqref{eq:algepsp02}, one finds that \eqref{eq:algepsp02}, after multiplication with $\tau^\mu$ reduces to
\begin{align}
  \left(\mathrm{Im}(u) + \mathrm{Re}(u) \gamma_0 \right) \left(\mathds{1}_2 - X^{-a} \gamma_a\right) \zeta_- = 0 \,.
\end{align}
If $u=0$, this equation is again identically satisfied for any $\zeta_-$. When $u\neq 0$, the matrix $\left(\mathrm{Im}(u) + \mathrm{Re}(u) \gamma_0 \right)$ is invertible and the above equation is equivalent to
\begin{align} \label{eq:zetaeigvector}
  X^{-a} \gamma_a \zeta_- = \zeta_- \,.
\end{align}
Since $X^{-a} \gamma_a$ is diagonalizable and has one eigenvalue $1$ and one eigenvalue $-1$, one sees that one can find one solution of this equation, given by an eigenvector with eigenvalue $1$. Note also that one can then recover \eqref{eq:defNXY} from \eqref{eq:zetaeigvector}, by multiplying both sides of \eqref{eq:zetaeigvector} from the left with $\bar{\zeta}_- \gamma_b$. In case  $u\neq 0$ and $\zeta_-$ obeys \eqref{eq:zetaeigvector}, plugging eqs. \eqref{eq:soltau0a1}, \eqref{eq:solv1} and \eqref{eq:zetaeigvector} into \eqref{eq:diffepsp0}, gives
\begin{align}
  D_\mu \zeta_- = \frac12 Y^{-a} D_\mu X^-_a \gamma_0 \zeta_- \,.
\end{align}
Similar manipulations show that this same equation also holds when $u=0$.
The spin connection terms in the covariant derivatives of this equation can be shown to cancel, so that one finds the following equation:
\begin{align}
  \partial_\mu \zeta_- = \frac12 Y^{-a} \partial_\mu X^-_a \gamma_0 \zeta_- = \frac12\left(X^-_2 \partial_\mu X^-_1 - X^-_1 \partial_\mu X^-_2\right)\gamma_0 \zeta_- \,.
\end{align}
This equation can be integrated to yield the solution\footnote{Here and in the following, we take the principal value of $\arctan$.}
\begin{align}
  \zeta_- = \exp\left(\frac12 \arctan{\left(\frac{X^-_1}{X^-_2}\right)}\gamma_0 \right) \zeta^-_0 \,,
\end{align}
where $\zeta^-_0$ is a constant spinor. For $u = 0$, this constant spinor is unconstrained, yielding two linearly independent solutions. In case $u\neq0$, $\zeta^-_0$ has to obey
\begin{align}
  \gamma_2 \zeta^-_0 = \mathrm{sign}(X^-_2) \zeta^-_0 \,,
\end{align}
to ensure that \eqref{eq:zetaeigvector} holds. One thus finds that there is only one solution when $u\neq 0$. In this way, we have shown that the conditions \eqref{eq:vistauab}--\eqref{eq:solv1} ensure that a solution of \eqref{eq:algepsp01}--\eqref{eq:diffepsp0} can be found. This solution is globally well-defined by virtue of the assumption that $X^-_a$ is globally well-defined, thus proving the theorem. 
\end{proof}
Note that we expressed the solution \eqref{eq:solv1} for $v_\mu$ in terms of the vectors $X_a^-$, $Y_a^-$, that are constructed from a Killing spinor. In case $u\neq 0$, this expression for $v_\mu$ is unambiguous, since there is only one solution $\zeta_-^{(1)}$ of the Killing spinor equations. In case $u=0$, there exist two independent Killing spinors $\zeta_-^{(1)}$ and $\gamma_0 \zeta_-^{(1)}$. Since there is no canonical choice of which Killing spinor to use to construct the vectors $X_a^-$, $Y_a^-$, one should make sure that the expression \eqref{eq:solv1} with $u=0$ does not depend on such a choice. This is indeed the case, as can be seen by taking an arbitrary linear combination
\begin{equation}
  \chi = a \zeta_-^{(1)} + b \gamma_0 \zeta_-^{(1)} \,, \qquad \qquad a, b \in \mathbb{R} \,,
\end{equation}
and defining
\begin{equation}
  N^\chi = \rmi \bar{\chi} \gamma_0 \chi \,, \qquad X^\chi_a = \frac{1}{N^\chi} \rmi \bar{\chi} \gamma_{0a} \chi \,, \qquad Y^\chi_a = \frac{1}{N^\chi} \rmi \bar{\chi} \gamma_a \chi \,.
\end{equation}
One finds that $X^\chi_a$ is still a unit vector and that moreover
\begin{equation}
  Y^{\chi a} D_\mu X_a^\chi = Y^{- a} D_\mu X_a^- \,,
\end{equation}
so that the expression \eqref{eq:solv1} for $v_\mu$ is indeed independent of the choice of Killing spinor, when $u=0$.

When put in a background that is subject to the relations \eqref{eq:vistauab}, \eqref{eq:soltau0a1}, and \eqref{eq:solv1}, the matter multiplet \eqref{eq:3dtrans} realizes a rigid superalgebra. The anti-commutator of the supercharges closes on bosonic symmetries of the theory, i.e. isometries and local Bargmann transformations. Let us denote the supercharges associated to solutions $(0,\zeta^{(i)}_-)$ ($i=1,2$) of the Killing spinor equations by $Q(\zeta^{(i)}_-)$. In case there is only one Killing spinor $(0,\zeta_-^{(1)})$, we find that $Q(\zeta^{(1)}_-)$ satisfies the following anti-commutation relation:
\begin{align}\label{eq:Q-Q-}
    \lc Q(\zeta^{(1)}_-),Q(\zeta^{(1)}_-)\rc = -\rmi\frac{\sqrt2}{2}\lr\delta_{\mathrm{U}(1)}(N^-) -\frac12\,\delta_G(\tau_{0a}N^-)\rr\,,
\end{align}
where $N^- = \rmi\bar\zeta_-^{(1)}\gamma_0\zeta_-^{(1)}$, as defined in eq. \eqref{eq:defNXY}. 
The transformation $\delta_{\mathrm{U}(1)}(N^-)$ corresponds to a central charge transformation with parameter $N^-$. This transformation was defined in eq. \eqref{eq:partnumber}. The transformation $\delta_G(\tau_{0a} N^-)$ corresponds to a local Galilean boost with parameter $\tau_{0a} N^-$. This boost acts non-trivially only on $\bar\pi\psi_-$ as follows
\begin{align}
  \delta_G(\tau_{0a} N^-) \bar \pi\psi_- &= -\frac{\rmi \sqrt{2}}{2} \tau_{0a} N^- \gamma^a \pi \psi_+ \,.
\end{align}
Let us now turn to the case, in which there is a second Killing spinor $(0,\zeta_-^{(2)})=(0,\gamma_0\zeta_-^{(1)})$. The anti-commutator of the supercharge $Q(\zeta_-^{(2)})$ with itself satisfies an anti-commutation relation that is formally the same as in eq. \eqref{eq:Q-Q-}. The mixed anti-commutator vanishes: $\lc Q(\zeta_-^{(1)}),Q(\zeta_-^{(2)})\rc =0$. Summarizing:
\begin{align}\label{eq:Q-Q-full}
    \lc Q(\zeta^{(i)}_-),Q(\zeta^{(j)}_-)\rc = -\rmi\,\delta^{ij}\,\frac{\sqrt2}{2}\lr\delta_{\mathrm{U}(1)}(N^-) -\frac12\,\delta_G(\tau_{0a}N^-)\rr\qquad \forall\, i,j=1,2\,.
\end{align}
Since the Killing spinors $(0,\zeta^{(i)}_-)$ do not carry $\mathrm{U}(1)$ charge and are inert under boosts, it is furthermore true that $\ls Q(\zeta^{(i)}_-),\lc Q(\zeta^{(j)}_-),Q(\zeta^{(k)}_-)\rc\rs=0$ (with $i,j,k=1,2$). The supercharges thus commute with the central charge symmetry and local Galilean boosts.

\subsection[The case $\zeta_+ \neq 0$]{The case \boldmath$\zeta_+ \neq 0$} \label{ssec:zetaplusneq0}

As mentioned in section \ref{ssec:KSeqs}, the Killing spinor equations \eqref{eq:OMalg1}--\eqref{eq:ucKSE2} are covariant with respect to local Galilean boosts. From \eqref{eq:boostepsilons} one sees that, in case $\zeta_+$ is not identically zero, one can completely fix this gauge freedom by setting $\zeta_- = 0$. Indeed, in case $\zeta_- \neq 0$ one can try to find a boost with parameters $\lambda^a$ such that 
\begin{align} \label{eq:boostfixeq}
  \zeta_- - \frac{1}{\sqrt{2}} \lambda^a \gamma_{a0} \zeta_+ = 0 \,, 
\end{align}
i.e. such that the boosted $\zeta_-$ is zero. Eq. \eqref{eq:boostfixeq} can be easily solved for the boost parameters $\lambda_a$ as follows
\begin{align}
  \lambda_a = \sqrt{2}\, \frac{\bar{\zeta}_+ \gamma_a \zeta_-}{\bar{\zeta}_+ \gamma_0 \zeta_+} \,.
\end{align}
Since $\bar{\zeta}_+ \gamma_0 \zeta_+ \propto \zeta_+^\dag \zeta_+ \neq 0$ for $\zeta_+ \neq 0$, this expression for the boost parameters is well-defined and one sees that one can indeed completely fix the boost gauge symmetry that the Killing spinor equations exhibit by setting $\zeta_- = 0$. In the following, we will assume that the boost gauge symmetry can be fixed in this way and we will look for solutions of the Killing spinor equations of the form $(\zeta_+, 0)$.\footnote{Strictly speaking, we are assuming here that $\zeta_+$ does not have any isolated zeros. In that case, one could not apply the boost gauge fixing $\zeta_- = 0$ at the positions of the zeros of $\zeta_+$. We will not discuss this possibility further here.}

Putting $\zeta_- = 0$ in the Killing spinor equations \eqref{eq:OMalg1}--\eqref{eq:ucKSE2} leads to the following equations:
\begin{subequations}
\begin{align}
  & \left(4 v + \epsilon^{ab} \tau_{ab} \right) \gamma_0 \zeta_+ = 0 \label{eq:algepspp1} \,, \\
  & \left(\frac32 \tau^{a0} \gamma_{a0} - \gamma^a v_a  + \mathrm{Re}(u) \gamma_0 - \mathrm{Im}(u)\right) \zeta_+ = 0 \label{eq:algepspp2} \,, \\
  & \left(\omega_\mu{}^a \gamma_{a0} + \frac13 e_\mu{}^a v_0 \gamma_a\right) \zeta_+ = 0 \label{eq:algepspp3} \,, \\
  & D_\mu \zeta_+ = \Big(-\frac14 \tau_{0\mu}  - \frac12 v_\mu \gamma_0  + \frac16 e_\mu{}^a v^b \gamma_a \gamma_b \gamma_0  + \frac13 \tau_\mu v_0 \gamma_0  \nonumber \\
  &\qquad \qquad \ \, - \frac16 \mathrm{Re}(u) e_\mu{}^a \gamma_a - \frac16 \mathrm{Im}(u) e_\mu{}^a \gamma_{a0} \Big)\zeta_+ \label{eq:diffepspp} \,,
\end{align}
\end{subequations}
where
\begin{align}
  D_\mu \zeta_+ = \partial_\mu \zeta_+ + \frac14 \omega_\mu{}^{ab} \gamma_{ab} \zeta_+ \,.
\end{align}
 Note that, in contrast to the previous case, the spin-connection field $\omega_\mu{}^a$ now also enters the equations.
We can again assume that any non-trivial solution for $\zeta_+$ of these equations is nowhere vanishing, via an argument analogous to the one given in section \ref{ssec:zetapluszero}. The conditions under which eqs. \eqref{eq:algepspp1}--\eqref{eq:diffepspp} admit non-trivial globally well-defined solutions can then be phrased as follows:
\begin{theorem}\label{thm2}
  The equations \eqref{eq:algepspp1}--\eqref{eq:diffepspp} have one non-trivial globally well-defined solution for $\zeta_+$ if and only if $\tau_{0\mu}$ is an exact one-form and there exists a globally well-defined unit vector $X^+_a$ such that the following conditions hold:{\footnote{Once again $v^{b}=e^{\mu b}v_{\mu}$ and $v_0=\tau^{\mu}v_{\mu}$, which appear on the right-hand side of conditions \eqref{eq:soltau0app} and \eqref{eq:solboostconn} respectively, are fully determined by condition \eqref{eq:solv1pp}. In view of \eqref{eq:solboostconn}, this means in particular that the connections for rotations and boosts are not independent for this class of solutions.}}
  \begin{subequations}
  \begin{align}
    \epsilon^{ab} \tau_{ab} &= -4 v \label{eq:vistauabpp} \,, \\
    \tau_{0a} &= \frac23 \left(\epsilon_{ab} v^b + \mathrm{Re}(u) X^+_a + \mathrm{Im}(u) Y^+_a \right) \label{eq:soltau0app} \,, \\
    \omega_\mu{}^a &= -\frac13 \epsilon^{ab} e_{\mu b} v_0 \label{eq:solboostconn} \,, \\
    v_\mu &= -3 \tau_\mu Y^{+a} D_0 X^+_a + \frac12 e_\mu{}^a \left(-3 Y^{+b} D_a X^+_b - \mathrm{Re}(u) Y^+_a + \mathrm{Im}(u) X^+_a \right) \label{eq:solv1pp} \,,
  \end{align}
  \end{subequations}
  where $Y^+_a = \epsilon_{ab} X^{+ b}$ and $D_\mu X^+_a = \partial_\mu X^+_a + \omega_{\mu a}{}^b X^+_b$. There are two independent globally well-defined solutions for $\zeta_+$ if and only if $\tau_{0\mu}$ is exact and there exists a globally well-defined unit vector $X^+_a$ such that the conditions \eqref{eq:vistauabpp}--\eqref{eq:solv1pp} hold with $u=0$.
\end{theorem}
\begin{proof}
The proof of this statement proceeds in an entirely similar fashion to the analogous theorem of section \ref{ssec:zetapluszero}. Let us thus first assume the existence of one non-trivial, globally well-defined solution $\zeta_+^{(1)}$ of eqs. \eqref{eq:algepspp1}--\eqref{eq:diffepspp} and show that this implies the conditions \eqref{eq:vistauabpp}--\eqref{eq:solv1pp}, as well as the exactness of $\tau_{0\mu}$. Via similar reasoning as in section \ref{ssec:zetapluszero}, it can be easily seen that eqs. \eqref{eq:vistauabpp} and \eqref{eq:solboostconn} follow when eqs. \eqref{eq:algepspp1} and \eqref{eq:algepspp3} are satisfied for a non-trivial $\zeta_+^{(1)}$. The existence of a nowhere vanishing and globally well-defined $\zeta_+^{(1)}$ allows us to define 
\begin{align} \label{eq:defNXYpp}
  N^+ = - \rmi \bar{\zeta}_+^{(1)} \gamma_0 \zeta_+^{(1)} \,, \qquad X^+_a = - \frac{1}{N^+} \rmi \bar{\zeta}_+^{(1)} \gamma_{0a} \zeta_+^{(1)} \,, \qquad Y^+_a = - \frac{1}{N^+} \rmi \bar{\zeta}_+^{(1)} \gamma_a \zeta_+^{(1)} \,.
\end{align}
As in section \ref{ssec:zetapluszero}, $N^+$ is nowhere vanishing because $\zeta_+^{(1)}$ is and the vectors $X^+_a$ and $Y^+_a$ are globally well-defined. By virtue of their definition and Fierz identities, they obey
\begin{align} \label{eq:propsXYpp}
  & X^+_a = - \epsilon_{ab} Y^+_b \,, \qquad \quad X^{+a} X_a^+ = 1 = Y^{+a} Y^+_a \,, \nonumber \\
  & X^{+a} \gamma_a \zeta_+^{(1)} = \zeta_+^{(1)} \,, \qquad \qquad Y^{+a} \gamma_a \zeta_+^{(1)} = \gamma_0 \zeta_+^{(1)} \,.
\end{align}
With the help of $X^+_a$ and $Y^+_a$, we can then rewrite \eqref{eq:algepspp2} as
\begin{align}
  \left(-\frac32 \epsilon^{ab} \tau_{0b} - v^a + \mathrm{Re}(u) Y^{+a} - \mathrm{Im}(u) X^{+a} \right) \gamma_a \zeta_+^{(1)} = 0 \,,
\end{align}
from which \eqref{eq:soltau0app} follows. Using \eqref{eq:soltau0app} as well as \eqref{eq:propsXYpp} in the differential condition \eqref{eq:diffepspp} on $\zeta_+^{(1)}$, we then find
\begin{align}
  & D_\mu \zeta_+^{(1)} = -\frac12 \tau_{0\mu} \zeta_+^{(1)} + C_\mu^+ \gamma_0 \zeta_+^{(1)} \,,  \\
  & \mathrm{where} \qquad C_\mu^+ = -\frac16\left(\tau_\mu v_0 + 2 e_\mu{}^a v_a + \mathrm{Re}(u) e_\mu{}^a Y^+_a - \mathrm{Im}(u) e_\mu{}^a X^+_a \right) \,.
\end{align}
From this equation, one derives that
\begin{align} \label{eq:covderNX}
  \partial_\mu \left(\log(N^+)\right) = - \tau_{0\mu} \,, \qquad \qquad D_\mu X^+_a = 2 C_\mu^+ Y^+_a \,.
\end{align}
From the second equation, one finds
\begin{align}
  C_\mu^+ = \frac12 Y^{+a} D_\mu X^+_a \,,
\end{align}
which can be rewritten as \eqref{eq:solv1pp}. Note that $\log(N^+)$ is well-defined, since $N^+$ is a well-defined function that is strictly positive. The first equation of \eqref{eq:covderNX} then says that $\tau_{0\mu}$ is an exact form.

In case there is a second solution $\zeta_+^{(2)}$ of eqs. \eqref{eq:algepspp1}--\eqref{eq:diffepspp}, we can follow a similar reasoning as in Theorem 1 to show that the conditions \eqref{eq:vistauabpp}--\eqref{eq:solv1pp} have to be satisfied with $u=0$. Indeed, as in Theorem 1, we can choose
\begin{equation}
  \zeta_+^{(2)} = \gamma_0 \zeta_+^{(1)} \,.
\end{equation}
Checking that this is a solution of eqs. \eqref{eq:algepspp1} and \eqref{eq:algepspp3} again leads to the conditions \eqref{eq:vistauabpp} and \eqref{eq:solboostconn}. One also finds that requiring that $\zeta_+^{(2)}$ is a solution of \eqref{eq:algepspp2} leads to
\begin{equation}
\tau_{0a} = \frac23 \left(\epsilon_{ab} v^b - \mathrm{Re}(u) X^+_a - \mathrm{Im}(u) Y^+_a \right) \,.
\end{equation}
Since this should hold simultaneously with \eqref{eq:soltau0app}, one finds that $u=0$ and that \eqref{eq:soltau0app} holds with $u=0$. One can then again show that the differential condition \eqref{eq:diffepspp}, with $u=0$ and  evaluated for $\zeta_+^{(1)}$ (or equivalently for $\zeta_+^{(2)}$), reduces to
\begin{align}
  & D_\mu \zeta_+^{(1)} = -\frac12 \tau_{0\mu} \zeta_+^{(1)} + c_\mu^+ \gamma_0 \zeta_+^{(1)} \,,  \\
  & \mathrm{where} \qquad c_\mu^+ = -\frac16\left(\tau_\mu v_0 + 2 e_\mu{}^a v_a  \right) \,,
\end{align}
from which exactness of $\tau_{0\mu}$ and eq. \eqref{eq:solv1pp} with $u=0$ can be derived as above. In this way, we see that the existence of a globally well-defined solution $\zeta_+^{(1)}$ of eqs. \eqref{eq:algepspp1}--\eqref{eq:diffepspp} implies exactness of $\tau_{0\mu}$ and the existence of a globally well-defined vector $X^+_a$ such that the conditions \eqref{eq:vistauabpp}--\eqref{eq:solv1pp} hold, where $u=0$ in case there are two independent solutions.

Let us now assume that $\tau_{0\mu}$ is exact and that one can find a globally well-defined vector $X^+_a$ such that eqs. \eqref{eq:vistauabpp}--\eqref{eq:solv1pp} are valid. One can then easily see that the Killing spinor equations \eqref{eq:algepspp1} and \eqref{eq:algepspp3} are identically satisfied for any $\zeta_+$, by virtue of \eqref{eq:vistauabpp} and \eqref{eq:solboostconn}. Plugging \eqref{eq:soltau0app} in \eqref{eq:algepspp2}, one finds that \eqref{eq:algepspp2} reduces to
\begin{align}
  \left(\mathrm{Re}(u) \gamma_0 - \mathrm{Im}(u) \right) \left(\mathds{1}_2 - X^{+a} \gamma_a\right) \zeta_+ = 0 \,.
\end{align}
When $u=0$, this equation is again identically satisfied for any $\zeta_+$. When $u\neq 0$, we can use the fact that then $\mathrm{Re}(u) \gamma_0 - \mathrm{Im}(u)$ is invertible to infer that $\zeta_+$ is an eigenvector of $X^{+a} \gamma_a$ with eigenvalue +1. Such an eigenvector can always be found, since $X^{+a} \gamma_a$ is diagonalizable with one eigenvalue +1 and the other eigenvalue -1. Finally, in this case, we can use \eqref{eq:soltau0app}, \eqref{eq:solv1pp} and the fact that $\zeta_+$ has to be an eigenvector of $X^{+a} \gamma_a$ with eigenvalue 1, in \eqref{eq:diffepspp} to find that \eqref{eq:diffepspp} reduces to
\begin{align}
  D_\mu \zeta_+ = -\frac12 \tau_{0\mu} \zeta_+ + \frac12 Y^{+a} D_\mu X^+_a \gamma_0 \zeta_+ \,.
\end{align}
Similar manipulations give the same equation when $u=0$.
The spin connection terms in the covariant derivatives of this equation again cancel out, leaving one with
\begin{align} \label{eq:finaldiffeqpp}
  \partial_\mu \zeta_+ = -\frac12 \tau_{0\mu} \zeta_+ + \frac12 Y^{+a} \partial_\mu X^+_a \gamma_0 \zeta_+ = -\frac12 \tau_{0\mu} \zeta_+ +  \frac12\left(X^+_2 \partial_\mu X^+_1 - X^+_1 \partial_\mu X^+_2\right)\gamma_0 \zeta_+ \,.
\end{align}
Exactness of $\tau_{0\mu}$ can now be invoked to write
\begin{align}
  -\frac12 \tau_{0\mu} = \partial_\mu \Phi \,,
\end{align}
where $\Phi$ is a well-defined function. The equation \eqref{eq:finaldiffeqpp} can then be integrated to
\begin{align}
  \zeta_+ = \rme^\Phi \exp\left(\frac12 \arctan{\left(\frac{X^+_1}{X^+_2}\right)}\gamma_0 \right) \zeta^+_0 \,,
\end{align}
where $\zeta^+_0$ is a constant spinor. When $u=0$, this constant spinor is unconstrained, leading to two independent solutions. When $u\neq 0$, $\zeta_0^+$ obeys
\begin{align}
 \gamma_2 \zeta^+_0 = \mathrm{sign}(X^+_2) \zeta^+_0 \,,
\end{align}
to ensure that $\zeta_+$ is an eigenvector of $X^{+a} \gamma_a$ of eigenvalue 1. This shows that there is only one solution when $u\neq 0$. In this way, we have shown that the conditions \eqref{eq:vistauabpp}--\eqref{eq:solv1pp} ensure that a solution of \eqref{eq:algepspp1}--\eqref{eq:diffepspp}  can be found. This solution is globally well-defined by virtue of the well-definedness of $X^+_a$ and $\Phi$, thus proving the theorem.
\end{proof}
As in theorem 1, one should show that the expression for $v_\mu$ is independent of the choice of Killing spinor when $u=0$. This can be done analogously to the discussion at the end of section \ref{ssec:zetapluszero}.

Let us finally comment on the rigid superalgebra that is obeyed by the matter multiplet \eqref{eq:3dtrans}, when placed in a background, in which $\tau_{0\mu}$ is exact and the relations \eqref{eq:vistauabpp}--\eqref{eq:solv1pp} hold. Let us denote the supercharges associated to a solution $(\zeta^{(i)}_+,0)$ ($i=1,2$) of the Killing spinor equations by $Q(\zeta^{(i)}_+)$. Considering first the case, in which there is only one Killing spinor $(\zeta_+^{(1)},0)$, we find the following anti-commutation relation
\begin{align}\label{eq:Q+Q+}
    \lc Q(\zeta^{(1)}_+),Q(\zeta^{(1)}_+)\rc  = -\rmi\,\frac{\sqrt2}{2}\,\mathcal{L}\ls N^+\tau^\mu\rs\,,
\end{align}
where $N^+ = -\rmi\bar\zeta_+^{(1)}\gamma_0\zeta_+^{(1)}$. The operator $\mathcal{L} \ls N^+ \tau^\mu \rs$ acts as an ordinary Lie derivative along $N^+ \tau^\mu$ on scalars and in the following way on fermions
\begin{align}
    \mathcal{L}\ls N^+\tau^\mu\rs \psi_\pm = N^+\tau^\mu\lr \bar\nabla_\mu\psi_\pm - \frac14 \tau_\mu\lr Y_c^+D_0X_c^+\epsilon^{ab}\rr\gamma_{ab}\psi_\pm\rr\,.
\end{align}
Note that the second term on the right-hand-side takes the form of a local rotation. 
Let us now assume that there exists a second Killing spinor $(\zeta_+^{(2)},0)$, with $\zeta_+^{(2)} = \gamma_0\zeta_+^{(1)}$. The anti-commutator $\{ Q(\zeta_+^{(2)}), Q(\zeta_+^{(2)})\}$ is then formally the same as in eq. \eqref{eq:Q+Q+}, whereas the mixed anti-commutator $\{ Q(\zeta_+^{(1)}), Q(\zeta_+^{(2)})\}$ is zero. Summarizing:
\begin{align}\label{eq:Q+Q+full}
    \lc Q(\zeta^{(i)}_+),Q(\zeta^{(j)}_+)\rc  = -\rmi\,\delta^{ij}\,\frac{\sqrt2}{2}\,\mathcal{L}\ls N^+\tau^\mu\rs\qquad\forall i,j=1,2\,.
\end{align}
Note that the Lie derivatives of the geometric background fields $\tau_\mu$, $e_\mu{}^a$ and $m_\mu$ along $N^+ \tau^\mu$, are zero up to local spatial rotations, Galilean boosts and central charge transformations (with parameters that depend on $N^+ \tau^\mu$), as can be checked by using equation \eqref{eq:covderNX}. In this sense, the quantity $N^+ \tau^\mu$ can be interpreted as a time-like Killing vector of the background Newton-Cartan geometry and the anti-commutation relation can be viewed as saying that the supercharges close into a time-like background isometry. This isometry furthermore commutes with the supercharges, i.e., $\ls Q(\zeta^{(i)}_+),\mathcal{L}\ls N^+\tau^\mu\rs\rs=0$.

\subsection[Cases with Killing spinors of both types $(0,\zeta_-)$ and $(\zeta_+,0)$]{Cases with Killing spinors of both types (\boldmath${0,\zeta_-}$) and (\boldmath${\zeta_+,0}$)}

Sections \ref{ssec:zetapluszero} and \ref{ssec:zetaplusneq0} dealt with the cases where there are one or two Killing spinors, that are either both of the type $(0,\zeta_-)$ (in section \ref{ssec:zetapluszero}) or of the type $(\zeta_+,0)$ (in section \ref{ssec:zetaplusneq0}). One can also consider cases where there are 2 or more Killing spinors of both types present. This can be done by combining the content of Theorems 1 and 2 of the previous two subsections. As an example, let us consider the constraints on the geometry and auxiliary fields in case there are four Killing spinors, i.e. in case there are two Killing spinors of the type $(0,\zeta_-)$ and two of the type $(\zeta_+,0)$. Theorems 1 and 2 with $u=0$ should then hold simultaneously. One then easily sees that
\begin{align}
  \epsilon^{ab} \tau_{ab} = \tau_{0a} = v = v_a = 0 \,.
\end{align}
There also exist well-defined unit vector fields $X^\pm_a$ (along with $Y^\pm_a = \epsilon_{ab} X^{\pm b}$) such that $v_\mu$ can be written in two different ways
\begin{align} \label{eq:vmutwoways}
   v_\mu &= \tau_\mu Y^{- a} D_0 X^-_a + \frac32 e_\mu{}^a Y^{- b} D_a X^-_b  \qquad \qquad \mathrm{and} \nonumber \\  v_\mu &= -3 \tau_\mu Y^{+ a} D_0 X^+_a - \frac32 e_\mu{}^a Y^{+ b} D_a X^+_b \,.
\end{align}
Extracting the $v_a$ components from these equations and requiring that they are zero, then implies that the spatial components $e^\mu{}_a \omega_\mu{}^{bc}$ of the rotation connection can be written in terms of $X^-_a$ as
\begin{align}
  e^\mu{}_a \omega_\mu{}^{bc} \epsilon_{bc} = -2 e^\mu{}_a \partial_\mu \left(\arctan\left(\frac{X^-_1}{X^-_2}\right)\right) \,,
\end{align}
and that the vector fields $X^\pm_a$ should obey the following constraint
\begin{align}
  e^\mu{}_a \partial_\mu \left(\arctan\left(\frac{X^-_1}{X^-_2}\right)\right) = e^\mu{}_a \partial_\mu \left(\arctan\left(\frac{X^+_1}{X^+_2}\right)\right) \,.
\end{align}
By looking at the time-like component $v_0$ of \eqref{eq:vmutwoways}, we see that the time-like component $\tau^\mu \omega_\mu{}^{ab}$ of the rotation connection and $v_0$ are given in terms of $X^\pm_a$ by
\begin{align}
  \tau^\mu \omega_\mu{}^{ab} \epsilon_{ab} &= - \frac12 \tau^\mu \partial_\mu \left(\arctan\left(\frac{X^-_1}{X^-_2}\right)\right) - \frac32 \tau^\mu \partial_\mu \left(\arctan\left(\frac{X^+_1}{X^+_2}\right)\right) \,, \nonumber \\
  v_0 &= \frac34 \tau^\mu \partial_\mu \left(\arctan\left(\frac{X^-_1}{X^-_2}\right) - \arctan\left(\frac{X^+_1}{X^+_2}\right) \right) \,.
\end{align}
We thus see that $\omega_\mu{}^{ab}$ is completely determined by $X^\pm_a$. The same is true for the boost connection $\omega_\mu{}^a$, since
\begin{align}
  \omega_\mu{}^a &= -\frac13 \epsilon^{ab} e_{\mu b} v_0 = - \frac14 \epsilon^{ab} e_{\mu b} \tau^\nu \partial_\nu \left(\arctan\left(\frac{X^-_1}{X^-_2}\right) - \arctan\left(\frac{X^+_1}{X^+_2}\right) \right) \,.
\end{align}

Let us now discuss the algebra that is realized when we consider Killing spinors of both types. We will again only consider the case in which there are four linearly independent Killing spinors of the form $(\zeta_+^{(i)},0)$ and $(0,\zeta_-^{(j)})$ (with $i,j=1,2$). The matter multiplet \eqref{eq:3dtrans} is then only non-minimally coupled through $v_0$. The supercharges denoted by $Q(\zeta^{(i)}_+)$ and $Q(\zeta^{(j)}_-)$ were discussed in sections \ref{ssec:zetapluszero} and \ref{ssec:zetaplusneq0} and as such, they satisfy relations \eqref{eq:Q-Q-full} (with $\tau_{0a}=0$) and \eqref{eq:Q+Q+full}. Moreover, one also finds that 
\begin{align}
    \lc Q(\zeta^{(1)}_+),Q(\zeta^{(1)}_-)\rc &= -\frac{\rmi}{2}\,\mathcal{L}\ls N^a e^{\mu}{}_a\rs\,,\notag\\
    \lc Q(\zeta^{(2)}_+),Q(\zeta^{(2)}_-)\rc &= +\frac{\rmi}{2}\,\mathcal{L}\ls N^a e^{\mu}{}_a\rs\,,\\
    \lc Q(\zeta^{(1)}_+),Q(\zeta^{(2)}_-)\rc &= -\frac{\rmi}{2}\,\mathcal{L}\ls \epsilon_{ab}\,N^a e^{\mu b}\rs\,,\notag
\end{align}
where $N^a = \rmi\bar\zeta^{(1)}_+\gamma_a\zeta^{(1)}_-$. The operators $\mathcal{L}\ls N^a e^{\mu}{}_a\rs$, resp. $\mathcal{L}\ls \epsilon_{ab}\,N^a e^{\mu b}\rs$ act as ordinary Lie derivative along $N^a e^{\mu}{}_a$, resp. $\epsilon_{ab}\,N^a e^{\mu b}$ on scalars. Their action on fermions includes an extra local boost term 
\begin{align}\label{eq:spatialLieLorentz}
    \mathcal{L}\ls N^a e^{\mu}{}_a\rs \psi_\pm = N^a e^{\mu}{}_a\bar\nabla_\mu\psi_\pm - \delta_G\lr \bar\nabla_0 N^a\rr \psi_\pm\,,
\end{align}
and analogously for $\mathcal{L}\ls \epsilon_{ab}\,N^a e^{\mu b}\rs$. Moreover, one can show that $N^a e^{\mu}{}_a$ and $\epsilon_{ab}N^a e^{\mu b}$ generate space-like isometries of the Newton-Cartan background geometry, in the sense discussed after eq. \eqref{eq:Q+Q+full}. Using that the Killing spinors are constant with respect to the derivative operation $\mathcal{L}$ of \eqref{eq:spatialLieLorentz}, one can furthermore show that the isometries commute with the supercharges
\begin{align}
    \ls Q(\zeta^{(i)}_\pm),\mathcal{L}[N^+\tau^\mu]\rs=0\,,\qquad\ls Q(\zeta^{(i)}_\pm),\mathcal{L}[N^a e^{\mu}{}_a]\rs=0\,, \qquad\ls Q(\zeta^{(i)}_\pm),\mathcal{L}[\epsilon_{ab}N^a e^{\mu b}]\rs=0\,.\notag
\end{align}

\subsection{Examples} \label{ssec:examples}

In this section, we will give two explicit classes of three-dimensional Newton-Cartan geometries that admit Killing spinors solving equations \eqref{eq:OMalg1}--\eqref{eq:ucKSE2}. The first class of examples is characterized by an integrable Newton-Cartan foliation structure (i.e.\,with $\tau_{[\mu} \partial_\nu \tau_{\rho]} = 0$), whereas the second class has a non-integrable foliation structure (i.e.\,with $\tau_{[\mu} \partial_\nu \tau_{\rho]} \neq 0$).

\paragraph{Integrable Foliation} In order to give our first class of examples, we split the coordinates $x^\mu$ as $x^\mu = \{x^0 = t, x^i\}$ ($i=1,2$) and choose an Ansatz that expresses the Newton-Cartan Vielbeine $\tau_\mu$, $e_\mu{}^a$ and the central charge gauge field $m_\mu$ in terms of three arbitrary functions $\kappa(t,x^i)$, $\lambda(t,x^i)$ and $\phi(t,x^i)$:
\begin{equation} \label{eq:confans}
    \tau_\mu \rmd x^\mu = \rme^\kappa\,dt\,,\qquad e_\mu{}^a \rmd x^\mu = \rme^\lambda \delta^a_i\,dx^i\,,\qquad\mathrm{and}\qquad m_\mu \rmd x^\mu = \phi\,\tau_\mu \rmd x^\mu\,.
  \end{equation}
For the projective inverse Vielbeine $\tau^\mu$, $e^\mu{}_a$, we then have
\begin{align}
    \partial_{0}= \, \tau^\mu \partial_\mu  = \rme^{-\kappa}\frac{\partial}{\partial t} \,,\qquad\qquad \partial_{a}=\, e^\mu{}_a \partial_\mu = \rme^{-\lambda} \delta_a^i \frac{\partial}{\partial x^i}\,.
\end{align}
This Ansatz is inspired by a class of four-dimensional Lorentzian backgrounds, discussed in \cite{Liu:2012bi}. Note that this Ansatz is such that the time-like Vielbein $\tau_\mu$ obeys the condition
\begin{align} \label{eq:hypsurface}
  \tau_{[\mu} \partial_\nu \tau_{\rho]} = 0 \,,
\end{align}
but is not closed ($\partial_{[\mu} \tau_{\nu]} \neq 0$). Specifically, the torsion $\tau_{\mu\nu} = 2 \partial_{[\mu}\tau_{\nu]}$ is found to be 
\be 
\tau_{a0}=\partial_a\kappa\,, \qquad \tau_{ab}=0\,.
\ee
Geometries that obey this condition have been encountered in the context of the AdS/CFT correspondence and are also commonly referred as `twistless torsional Newton-Cartan geometries' in the literature \cite{Christensen:2013lma,Christensen:2013rfa}. Geometrically, the condition \eqref{eq:hypsurface} says that the Newton-Cartan manifold is foliated in a one-dimensional time-direction and two-dimensional spatial slices. Our Ansatz then also specifies that the metric $h^{\mu\nu} = e^\mu{}_a e^{\nu a}$ on these spatial slices is conformally flat. The central charge gauge field $m_\mu$ in our Ansatz only has a non-zero time-like component $\tau^\mu m_\mu$, given by $\phi$. In the context of Newton-Cartan gravity (which has zero torsion $\tau_{\mu\nu} = 0$), $\phi$ would correspond to the Newton potential and we will refer to $\phi$ as the Newton potential here as well. Note that this is a slight abuse of terminology however, since the backgrounds we are looking for do not necessarily have to solve any equations of motion of an underlying non-relativistic gravitational theory.
  
We will first restrict ourselves to the case in which we are looking for backgrounds that have one or two Killing spinors of the form $(0,\zeta_-)$ and discuss the $(\zeta_+,0)$ case later. The background (auxiliary) fields then have to obey equations \eqref{eq:vistauab}--\eqref{eq:solv1}. Here, we will work out the ensuing relations between the background fields $\{v,\,v_\mu,\,u\}$ and the geometric data $\{\phi,\,\kappa,\,\lambda\}$ explicitly. Using \eqref{eq:defNCconnstaumunu}, it is straightforward to show that
\begin{align}\label{spinconnectionsexample1}
    \omega_\mu &\equiv \omega_\mu{}^{ab} \epsilon_{ab} =  2\,e_\mu{}^a\epsilon_{ab}\partial^b \lambda\,,\qquad \qquad \omega_\mu{}^a =  e_\mu{}^a \partial_0  \lambda -\tau_\mu\left(\partial^a\phi+ \phi\,\partial^a\kappa\right)\,.
\end{align}
The associated curvatures can be calculated from \eqref{eq:Jcurv} and \eqref{eq:Gcurv}. For this particular example, we will assume that $\partial_0\kappa=0=\partial_0\lambda$, so that the Newton-Cartan Vielbeine $\tau_\mu$ and $e_\mu{}^a$ are time-independent. With this assumption, we find the following expressions for the curvature of spatial rotations and boosts
\begin{align}
 \label{curvature ex1}   R_{\mu\nu}(J) &= -2\lr \partial_c\partial_c \lambda + \partial_c \lambda \partial_c \lambda \rr\,\epsilon_{ab}e_\mu{}^a e_\nu{}^b\,,\\
    R_{\mu\nu}(G^a) &= 2\lr \partial^b\Phi^a + \partial^b\kappa\,\Phi^a  - \partial^a \lambda\Phi^b -\Phi^{(a}\tau^{b)0}\rr\tau_{[\mu}e_{\nu] b} + 2\Phi^b\partial_b \lambda \,\tau_{[\mu}e_{\nu]}{}^a\,,
\end{align}
where $\Phi^a = \partial^a\phi + \phi\,\partial^a\kappa$. Note that $R_{\mu\nu}(J)$ does not depend on the Newton potential $\phi$. Moreover, the above expression for $R_{\mu\nu}(J)$ captures the curvature of the spatial slices. 

Demanding that there exist one or two solutions of the Killing spinor equations of the form $(\zeta_+=0, \,\zeta_-)$, entails requiring that there exists a well-defined unit vector field $X_a^-$, as outlined in section \ref{ssec:zetapluszero}. Here, we will choose this unit vector field to be constant and given by $X_1^-=X_2^-=1/\sqrt2$. Note that such a choice can be viewed as a gauge fixing for local spatial rotations. The relations \eqref{eq:vistauab}--\eqref{eq:solv1} then allow us to give explicit expressions for the background fields:
\begin{align}
    v&=0\,,\qquad v_0=0\,, \qquad v^a= -\frac12 \epsilon^{ab}\partial_b\kappa + \epsilon^{ab}\partial_b \lambda \,, \notag\\
    \mathrm{Re}(u)&= -X_a^-\partial^a\lr \lambda+\kappa\rr\,,\qquad 
    \mathrm{Im}(u) = Y_a^-\partial^a\lr \lambda+\kappa\rr\,,
\end{align}
and for the Killing spinor
\begin{align}
    \zeta_{-} = e^{\frac{\pi}{8}\gamma_0} \zeta_{0}^{-}\,,
\end{align}
where the constant spinor $\zeta_0^-$ is unconstrained if $u=0$ (corresponding to the case in which there are two Killing spinors) and obeys $\gamma_2\zeta_0^- = \zeta_0^-$ if $u\neq0$ (corresponding to the case in which there is only one Killing spinor). We note in passing that a straightforward calculation confirms that substitution of the above expressions for the background fields in the integrability condition \eqref{ICzeta-} indeed results in \eqref{curvature ex1}, as it should. 

As an illustrative example we consider a manifold whose spatial slices are isomorphic to the Poincar\'e disc with radius $\ell$:
\begin{align}
    \rme^{\lambda} = \frac{2}{1-x_ix^i/\ell^2}\,,\qquad\mathrm{where}\qquad x^ix^i< \ell^2\,, \quad i=1,2\,.\label{eq:PoincDisc}
\end{align}
When it comes to curvatures, we observe that $R_{\mu\nu}(J) = -2/\ell^2\,\epsilon_{ab}e_\mu{}^a e_\nu{}^b$ is completely determined by the Ansatz \eqref{eq:PoincDisc}. Note, however, that the Newton potential $\phi$ is unconstrained by supersymmetry---and thus also $R_{\mu\nu}(G^a)$.

We will focus on the case of two supercharges, which imposes that $u=0$ and thus $-\partial^a\kappa = \partial^a \lambda$. This leaves us with just one non-vanishing background field $v^a = -3/2\,\epsilon^{ab}\tau_{b0}$. Inserting the explicit expression
\begin{align}
     v^a =  \frac{3}{2}\lr \frac{x^2}{\ell^2}, -\frac{x^1}{\ell^2}\rr\label{eq:vaPoincDisc}
\end{align}
into the action \eqref{eq:3dact} and the transformation rules \eqref{eq:3dtrans} yields a supersymmetric theory that is coupled to the background \eqref{eq:PoincDisc}. Note that this leads to a number of non-minimal coupling terms that are suppressed by $1/\ell^2$, such that the action and the supersymmetry rules reduce to the flat space expressions in the limit $\ell\to\infty$.

To summarize, we have shown that a twistless torsional Newton-Cartan geometry of the form \eqref{eq:confans} (with $\tau_\mu$, $e_\mu{}^a$ time-independent) with Poincar\'e disc spatial slices (determined by \eqref{eq:PoincDisc}) and arbitrary Newton potential $\phi$ allows for two supercharges of the form $(0,\,\zeta_-)$.

Let us now turn to the $(\zeta_+,0)$ case and assume there are two supercharges as in the Poincar\'e disc example above. The relevant conditions now are those of Theorem \ref{thm2}, namely \eqref{eq:vistauabpp}--\eqref{eq:solv1pp}. Then the explicit expressions for the background fields, with unit vector being once more $X_1^+=X_2^+=1/\sqrt2$, are 
\bea 
v=0\,, \quad v_0=0\,,\quad  u=0\, \quad \text{and}\quad v^{a}=\frac 32 \epsilon^{ab}\tau_{b0}\,,
\eea 
while the Killing spinor takes the form 
\be \zeta_{+} = e^{\frac{\kappa}{2}+\frac{\pi}{8}\gamma_0} \zeta_{0}^{+}\,.
\ee
Note the sign difference in $v^{a}$ with respect to the previous example. In order to establish that this background indeed satisfies all conditions of Theorem \ref{thm2}, one should also examine the consistency of eqs. \eqref{eq:solboostconn} and \eqref{eq:solv1pp}, noting that the spin connection is given by \eqref{spinconnectionsexample1}. Consistency requires that $\Phi^{a}=0$, which can be solved by $\phi=e^{-\kappa}$. This in turn means that the central charge gauge field is a constant one-form for this solution, i.e. $m_{\mu}d x^{\m}= dt$. Note that the curvature of boosts vanishes in this case, $R_{\m\n}(G^{a})=0$. Thus, we have demonstrated that a twistless-torsional Newton-Cartan geometry with Poincar\'e disc spatial slices allows for two supercharges of the form $(\zeta_+,\,0)$ for the special potential $\phi=e^{-\kappa}$.

Finally, let us comment for completeness on the case of vanishing torsion as a special case of this class of examples. This is obtained for constant $\kappa$, which may be taken to be zero without loss of generality. With $\phi=0$, the Newton-Cartan fields are
\begin{equation} \label{eq:zerotorsion}
\tau_\mu \rmd x^\mu = dt\,,\qquad e_\mu{}^a \rmd x^\mu = \rme^\lambda \delta^a_i\,dx^i\,,\qquad\mathrm{and}\qquad m_\mu \rmd x^\mu = 0\,,
\end{equation}
for which indeed $\tau_{\mu\nu}=0$, corresponding to $\mathbb{R}\times {\cal M}_2$ with ${\cal M}_2$ an arbitrary 2-manifold. For time-independent $\lambda$, the boost connection vanishes, namely $\omega_{\mu}{}^{a}=0$, and the nonvanishing components of $\omega_{\mu}$ are $\omega_{a}=2\epsilon_{ab}\partial^b\lambda$. Curvature thus resides in the spatial slices. The background fields of the solution are then given as 
\begin{align}
v&=0\,,\qquad v_0=0\,,\qquad v_a= \frac 12 \omega_a\,,\notag\\
\mathrm{Re}(u)&= -\frac 12 Y^{-}_{a}\omega^a\,,\qquad
\mathrm{Im}(u) =-\frac 12 X^{-}_a\omega^a\,.
\end{align}
Such solutions descend from four-dimensional supersymmetric backgrounds of the form $\mathbb{R}^{1,1}\times {\cal M}_2$ upon null reduction. 

\paragraph{Non-Integrable Foliation} As a complementary class of examples, we consider torsional Newton-Cartan geometries that have
\begin{align}
    \tau_{[\mu}\partial_\nu\tau_{\rho]}\neq 0
\end{align}
which corresponds to a foliation of spacetime that is non-integrable. Note that such spacetimes are typically excluded since non-relativistic causality is violated, see e.g. \cite{Bergshoeff:2017dqq}. Here, however, we do not consider geometry as a physical spacetime, but rather as a rigid background. This kind of geometrical structure is well-studied in the mathematical literature and known as a contact structure, see e.g. \cite{mcduff_introduction_2017}.

We will now solve the Killing spinor equations explicitly by splitting the coordinates $x^\mu$ as $x^\mu=\{x^0=t,x^i\}$ $(i=1,2)$ and choosing the Ansatz
\begin{align} \label{eq:secondans}
    \tau_\mu\,dx^\mu = dt + \alpha_i\,dx^i\,,\qquad e_\mu{}^a dx^\mu = e_i{}^a dx^i\,,\qquad\mathrm{and}\qquad m_\mu\,dx^\mu = -\frac12\,\alpha_i\,dx^i\,,
\end{align}
which corresponds to the following expressions for the projective inverse Vielbeine $\tau^\mu,\,e^\mu{}_a$
\begin{align}
    \tau^\mu\partial_\mu = \frac{\partial}{\partial t}\,,\qquad e^\mu{}_a\partial_\mu = e^i{}_a\lr \frac{\partial}{\partial x^i} - \alpha_i\frac{\partial}{\partial t}\rr\,,
\end{align}
where $e^i{}_a$ is the matrix inverse of $e_i{}^a$. For simplicity, we assume that $\partial_t\alpha_i=0$ and $\partial_t e_i{}^a=0$, which leads to $\tau_{a0}=0$ but $\tau_{ab}=2e^{i}{}_ae^{j}{}_b\partial_{[i}\alpha_{j]}\neq 0$. In the context of contact geometry, this identifies $\tau^\mu\partial_\mu$ as the Reeb-vector field \cite{mcduff_introduction_2017}. Furthermore, the Newton-Cartan spin connections may be once again computed using eqs. \eqref{eq:defNCconnstaumunu} and we find 
\bea 
\omega_{\m}{^{a}}&=&\frac 14 e_{\m b}\tau^{ab}\,,\\ 
\omega_{\m}{}^{ab}&=&2e^{\n[a}\partial_{[\m}e_{\n]}{}^{b]}-e_{\m}{}^{c}e^{\n a}e^{\rho b}\partial_{[\n}e_{\rho]c}+\frac 14 \tau_{\m}\tau^{ab}\,.
\eea
We observe that the two spin connections are related as $\omega_{\m}{}^{a}=e_{\m b}\tau^{\n}\omega_{\n}{}^{ab}$.

Backgrounds with non-integrable foliation $\tau_{ab}\neq 0$ cannot preserve more than two supercharges. 
 This can be seen by noting that having  Killing spinors of both the
form $(\zeta_+,0)$ and $(0,\zeta_-)$ implies that both eqs.\,\eqref{eq:vistauab} and \eqref{eq:vistauabpp} hold, implying that $\tau_{ab}=0$. Hence we conclude that backgrounds with $\tau_{ab}\neq 0$ can only preserve supercharges of the form $(0,\zeta_-)$ or $(\zeta_+,0)$---and thus at most two supercharges. In our present analysis, we assume two Killing spinors of the first type. 
	
 An explicit Newton-Cartan geometry for this example is  
\begin{align}\label{eq:RxSO(3)}
\alpha_i\,dx^i= - \frac{\ell}{2}\,\cos\eta_1\,d\eta_2\,,\qquad e_i{}^1 dx^i = \frac{\ell}{2}\,d\eta_1\,,\qquad e_i{}^2dx^i = \frac{\ell}{2}\,\sin\eta_1\,d\eta_2 \,,
\end{align}
defined on a three-manifold with coordinates
\begin{align}
x^\mu=(t,\eta_1,\eta_2) \qquad \qquad  t \in\mathds{R}, \ \ \eta_1\in [0,\pi), \ \ \eta_2\in [0,2\pi) \,.
\end{align}
In this case, the torsion is found to be $\tau_{ab}= 2 \epsilon_{ab}/\ell$. Moreover, the components of the Newton-Cartan spin connections can be written explicitly as 
\be 
\omega_\mu = \frac 1{\ell}\tau_{\mu}-\frac 4{\ell}\cot\eta_1\,e_\mu{}^2\,\quad \text{and}\quad   \omega_{\m}{}^{a}=\frac 1{2\ell}e_{\m b}\epsilon^{ab}\,.
\ee 
From the general discussion in section \ref{ssec:zetapluszero}, we initially conclude that
\begin{align}
    u=0=v^a \qquad \text{and} \qquad  v =  \frac3{\ell} \,.
\end{align}
Furthermore, we observe that a constant unit vector $X^{-a}$ is not a consistent choice in the present case, since it would render condition \eqref{eq:solv1} inconsistent.
Instead, we make the time-dependent choice
\begin{align}
  X^{-}_1=\sin \left(\sfrac{2t}{\ell}\right) \,, \qquad \qquad  X^{-}_2=\cos\left(\sfrac {2t}{\ell}\right) \,.
\end{align}
Then, condition \eqref{eq:solv1} implies that 
\be 
v_0=\frac 5{2\ell}\,.
\ee 
It is then straightforward to check that all conditions of Theorem \ref{thm1} are satisfied. Furthermore, the only nonvanishing components of the corresponding curvatures are found to be  
\begin{align}
R_{\mu\nu}(J) &= \frac{14}{\ell^2} \epsilon^{ab}e_\mu{}^a e_\nu{}^b\,, \qquad \qquad
 R_{\mu\nu}(G^a) = -\frac{1}{2\ell^2}\tau_{[\mu}e_{\nu]}{}^a\,.
\end{align}
This geometry has an $\mathds{R}\times SO(3)$ isometry, where the non-compact part corresponds to translations in time, as explained in \cite{Grosvenor:2017dfs}. In the same reference it is also shown that the geometry \eqref{eq:RxSO(3)} can be obtained from a null reduction of $\mathds{R}\times\mathds{S}^3$. The lighlike isometry is a linear combination of time-translations and translations along the Hopf fiber. The relativistic background $\mathds{R}\times\mathds{S}^3$ is known \cite{Festuccia:2011ws} to preserve four supercharges. Hence it is not surprising that the null reduction \eqref{eq:RxSO(3)} provides a supersymmetric background too. However, we observe that the three-dimensional geometry allows for just two supercharges of the form $(0,\zeta_-)$ with
\begin{align}
    \zeta_- (t) =e^{t\gamma_0/\ell} \zeta_0^-\,, \label{eq:explKSzetap}
\end{align}
where $\zeta_0^-$ is a constant, but otherwise unconstrained spinor.

Note that two supercharges have been lost by reducing to three dimensions. This can be traced back to the fact that not all Killing spinors in four dimensions can be made independent of the chosen lightlike isometry.\footnote{We thank Guido Festuccia for bringing that point to our attention and for remarking that an analogous reduction that does not break supersymmetry might be possible within the framework of New Minimal supergravity, see \cite{Festuccia:2011ws}.}

To summarize, we find a supersymmetric theory with two supercharges of the form $(0,\zeta_-)$ that is non-minimally coupled to the background \eqref{eq:secondans} (with $\tau_\mu$ and $e_\mu{}^a$ time-independent) with \eqref{eq:RxSO(3)}. All the terms that go beyond the flat space expression are suppressed as $\ell\to \infty$.

\section{Conclusions}

Recent years have witnessed considerable progress in understanding non-perturbative aspects of QFT by constructing and studying susy QFTs on curved backgrounds. Observables in such theories can often be calculated exactly via localization techniques \cite{Pestun:2016zxk}. This has led to new insights in the non-perturbative structure of susy QFTs and allowed e.g. precision tests of AdS/CFT and holography. Most studies in the literature are concerned with relativistic susy QFTs, i.e. susy QFTs on backgrounds in which local inertial frames are connected via relativistic space-time symmetries. A natural question to ask is whether similar techniques can be used to study the non-perturbative behaviour of non-relativistic susy QFTs that live on space-times, whose local inertial frames are connected by non-relativistic space-time symmetries, such as Newton-Cartan manifolds.

Motivated by this question, we have constructed examples of non-relativistic rigidly supersymmetric field theories on curved three-dimensional Newton-Cartan manifolds. We have in particular obtained a Lagrangian and supersymmetry transformation rules that describe the dynamics of a three-dimensional non-relativistic `pseudo-chiral' multiplet in curved non-relativistic backgrounds. The backgrounds are specified by a set of fields that determine the Newton-Cartan geometry, as well as by a set of auxiliary fields. The dynamical fields of the pseudo-chiral multiplet couple both minimally and non-minimally to the background fields. The Lagrangian is only supersymmetric when non-trivial supersymmetry parameters can be found as well-defined solutions of a set of non-relativistic Killing spinor equations. The latter are a set of algebraic and first order partial differential equations for the supersymmetry parameters that depend on the background geometric and auxiliary fields. We have derived necessary and sufficient conditions on the background fields for well-defined solutions of the Killing spinor equations to exist. Non-relativistic supersymmetric field theories can then be written down on Newton-Cartan manifolds that obey these conditions and we have given explicit examples of such backgrounds.

Here, we have obtained non-relativistic susy QFTs on non-trivial backgrounds via a null reduction of relativistic theories.\footnote{More specifically, we employed a twisted null reduction to obtain theories whose field equations are Schr\"odinger equations in curved backgrounds.} We have in particular reduced the Lagrangian, supersymmetry transformation rules and Killing spinor equations of a four-dimensional theory for a chiral multiplet in a relativistic background. The latter theory can be elegantly obtained by taking a rigid limit of matter coupled Lorentzian $\mathcal{N}=1$, $d=4$ off-shell Old Minimal supergravity \cite{Festuccia:2011ws}. Supersymmetry then requires that the allowed four-dimensional backgrounds, on which susy QFTs can be formulated, have a null Killing vector. The latter observation led us to consider the null reduction of four-dimensional relativistic susy QFTs on curved manifolds as a means to obtain three-dimensional non-relativistic ones.

One limitation of this method is that it only leads to non-relativistic susy QFTs that have a four-dimensional, relativistic origin and that it thus most likely does not lead to the most general class of non-relativistic susy QFTs on curved three-dimensional manifolds. Nevertheless, since the three-dimensional theories and backgrounds are qualitatively very different from their four-dimensional parents, it is useful to study them on their own and to give a purely three-dimensional analysis of which backgrounds lead to well-defined solutions of the non-relativistic Killing spinor equations. In this way, we can extract some features of non-relativistic susy QFTs on curved backgrounds that can be expected to hold regardless of whether they are obtained via null reduction or via other means. For instance, we have noticed in this paper that the set of non-relativistic Killing spinor equations consists not only of first order partial differential equations, but also contains algebraic equations. The latter do not decouple, but rather arise by transforming the differential Killing spinor equations under local Galilean boosts. In this way, the non-relativistic Killing spinor equations form a reducible indecomposable representation of local Galilei symmetries. Since such reducible indecomposable representations are common in non-relativistic theories, we expect that the appearance of algebraic Killing spinor equations is not just a consequence of the null reduction, but is a generic feature of non-relativistic susy QFTs on curved backgrounds.

There are several ways in which the work presented here can be extended. In this paper, we focused on the null reduction of four-dimensional $\mathcal{N}=1$ susy QFTs on curved backgrounds that involve a single chiral multiplet with canonical K\"ahler potential and zero superpotential. This reduction can straightforwardly be extended to field theories with a more general matter content and couplings or to theories with extended supersymmetry. We also focused on the reduction of four-dimensional theories that were obtained as a rigid limit of supersymmetric matter field theories coupled to Old Minimal off-shell supergravity. One can also consider theories that are obtained from a rigid limit of four-dimensional theories that are coupled to New Minimal off-shell supergravity. Such theories admit an $R$-symmetry and it would be interesting to consider their null reduction. Another direction one can consider concerns the construction of superconformal theories in curved Newton-Cartan backgrounds. Recently, non-relativistic superconformal theories have been studied in the context of the six-dimensional $(2,0)$ theory \cite{Lambert:2018lgt,Lambert:2019fne}. Such theories could also be used to give supersymmetric extensions of work on the anomaly structure of non-relativistic scale-invariant field theories \cite{jensen:2018an,Arav:2016xjc,Auzzi:2015fgg,Auzzi:2016lxb,Auzzi:2016lrq,Pal:2017ntk,Pal:2016rpz}. It would also be of interest to see whether it is possible to find interesting non-relativistic susy QFTs in the class of theories that can be obtained via null reduction, whose non-perturbative dynamics can be studied using localization techniques.

An important open question concerns the possibility of obtaining non-relativistic susy QFTs on curved backgrounds without using the null reduction, so without relying on higher-dimensional relativistic results. As mentioned above, the null reduction most likely does not lead to the most general non-relativistic susy QFTs on curved backgrounds. Moreover, there are also non-relativistic geometries that are characterized by a different foliation structure than ordinary Newton-Cartan geometry, such as e.g. string Newton-Cartan geometries that admit a foliation with spatial leaves of codimension 2 \cite{Andringa:2012uz}. These string Newton-Cartan geometries are the ones non-relativistic strings \cite{Gomis:2000bd} naturally live in \cite{Bergshoeff:2018yvt,Bergshoeff:2019pij}. They can not be obtained from null reduction and one will thus have to resort to different techniques to construct susy QFTs on them. It is therefore an interesting question to see whether non-relativistic susy QFTs on curved backgrounds can be constructed more directly, without having to rely on relativistic results. In this regard, it would be interesting to see whether one can gain more insight into the structure of non-relativistic Killing superalgebras and Killing spinor equations via cohomological methods \cite{Figueroa-OFarrill:2015rfh,deMedeiros:2016srz,Figueroa-OFarrill:2016khp,deMedeiros:2018ooy}. In order to also be able to construct Lagrangians and supersymmetry transformation rules, one can consider mimicking the Festuccia-Seiberg method directly in three dimensions. This would involve taking a rigid limit of non-relativistic supersymmetric field theories coupled to off-shell supergravity. At present however, not much is known about non-relativistic, three-dimensional off-shell supergravity nor about matter couplings in non-relativistic supergravity. Partial results, based on a non-relativistic extension of superconformal tensor calculus, led to a three-dimensional supergravity multiplet, on which the superalgebra closes upon using only geometric constraints \cite{Bergshoeff:2015ija}. So far, a non-relativistic supergravity multiplet that realizes the underlying superalgebra without having to impose any constraints, has not been constructed. Matter couplings in non-relativistic supergravity have also not been studied yet, neither in the on-shell Newton-Cartan supergravity of \cite{Andringa:2013mma}, nor in the partially off-shell formulations with geometric constraints of \cite{Bergshoeff:2015ija}. In view of obtaining more general non-relativistic susy QFTs in curved backgrounds, studying possible off-shell formulations of non-relativistic supergravity and their matter couplings is clearly a very interesting and pressing problem. We hope to report more on this in the future.

\paragraph{Acknowledgements.}  Part of this work was done while E.A.B. was visiting the 
Erwin Schr\"odinger Int. Institute for Mathematics and Physics, University of Vienna as a Senior Research Fellow. We would like to thank Shira Chapman, Guido Festuccia, Jos\'e Figueroa-O'Farrill, Neil Lambert, Fabrizio Nieri, Yaron Oz, Silvia Penati and Avia Raviv-Moshe for stimulating discussions on the manuscript. This work is supported by the Croatian Science Foundation Project ``New Geometries for Gravity and Spacetime'' (IP-2018-01-7615), and also partially supported by the European Union through the European Regional Development Fund - The Competitiveness and Cohesion Operational Programme (KK.01.1.1.06).

\appendix

\section{Conventions}

In this paper, we use the mostly plus signature. The four-dimensional flat metric in Minkowski coordinates $\{x^0, x^1, x^2, x^3\}$ is thus given by $\mathrm{diag}(-1,1,1,1)$. Four-dimensional flat indices are denoted by $A$, $B$, $\cdots$, whereas four-dimensional curved indices are denoted by $M$, $N$, $\cdots$. Flat null coordinates $x^\pm$ are introduced via
\begin{equation}
  x^- = \frac{1}{\sqrt{2}} \left(x^0 - x^3 \right) \,, \qquad x^+ = \frac{1}{\sqrt{2}} \left(x^0 + x^3 \right) \,,
\end{equation}
so that the four-dimensional flat metric in null coordinates $\{x^-, x^+, x^1, x^2\}$ is given by
\begin{equation}
  \label{eq:flatmetric}
  \eta_{AB} = \bordermatrix{& - & + & b\cr
- & 0 & -1 & 0 \cr
+ & -1 & 0 & 0 \cr
a & 0 & 0 & \delta_{ab}} \,.
\end{equation}
The fully anti-symmetric Levi-Civita symbol $\epsilon_{ABCD}$ in flat indices is taken as
\begin{equation}
  \epsilon_{0123} = 1 \,, \qquad \epsilon^{0123} = -1 \,.
\end{equation}
Introducing $\epsilon_{ab}$ and $\epsilon^{ab}$ ($a=1,2$) with
\begin{equation}
  \epsilon_{12} = \epsilon^{12} = 1 \,,
\end{equation}
the Levi-Civita symbol in flat null indices is given by
\begin{equation}
  \epsilon_{-+ab} = \epsilon_{ab} \,, \qquad \epsilon^{-+ab} = - \epsilon^{ab} \,.
\end{equation}
The symbol $\epsilon_{MNOP}$ with curved indices is defined as the tensor
\begin{equation}
  \epsilon_{MNOP} = E_M{}^A E_N{}^B E_O{}^C E_P{}^D \epsilon_{ABCD} \,.
\end{equation}
Since this is a tensor and not a tensor density, we can freely raise and lower its indices with the metric. The spin connection of General Relativity is given by
\begin{align}
  \label{eq:spinconn}
  \Omega_M{}^{AB} = 2 E^{N[A} \partial_{[M} E_{N]}{}^{B]} - E^{N A} E^{R B} E_{M C} \partial_{[N} E_{R]}{}^C \,.
\end{align}

For four-dimensional spinors, we adopt the following conventions. Four-dimensional Gamma-matrices are denoted by $\Gamma_A$ and obey the Clifford algebra
\begin{equation} \label{eq:4dCliff}
  \{\Gamma_A, \Gamma_B\} = 2\, \eta_{AB}\, \mathds{1}_4 \,.
\end{equation}
The matrices $\Gamma^\pm$ in flat null indices are given in terms of $\Gamma^0$, $\Gamma^3$ in Minkowski indices by
\begin{equation}
  \label{eq:defGammapm}
  \Gamma^\pm = \frac{1}{\sqrt{2}} \left( \Gamma^0 \pm \Gamma^3 \right) \,.
\end{equation}
The charge conjugation matrix obeys
\begin{align}
  C^T = - C \,, \qquad \Gamma_A^T = - C \Gamma_A C^{-1} \,.
\end{align}
The matrix $\Gamma_5$ is defined as
\begin{equation}
  \label{eq:defGamma5}
  \Gamma_5 = \rmi \Gamma_0 \Gamma_1 \Gamma_2 \Gamma_3 \,,
\end{equation}
and obeys $\Gamma_5^\dag = \Gamma_5$ and $\Gamma_5^T = C \Gamma_5 C^{-1}$. A Majorana spinor $\epsilon$ is a spinor for which the Dirac conjugate $\rmi \epsilon^\dag \Gamma^0$ is equal to the Majorana conjugate $\epsilon^T C$. The subscript $L/R$ on a Majorana spinor $\epsilon$ denotes a chiral projection:
\begin{align}
  \epsilon_{L/R} = P_{L/R} \epsilon \,, \qquad  \qquad P_{L/R} = \frac12 \left(\mathds{1}_4 \pm \Gamma_5\right) \,.
\end{align}
The chiral projections of a Majorana spinor $\epsilon$ then obey
\begin{align}
  \epsilon_L^* = \rmi C \Gamma^0 \epsilon_R \,, \qquad \epsilon_R^* = \rmi C \Gamma^0 \epsilon_L \,.
\end{align}
A bar on a Majorana spinor denotes the Dirac or Majorana conjugate without ambiguity. For chiral projections of a Majorana spinor $\epsilon$, we adopt the following notation:
\begin{align}
  \bar{\epsilon}_L &= \frac12 \bar{\epsilon} \left(\mathds{1}_4 + \Gamma_5\right) \,, \qquad \bar{\epsilon}_R = \frac12 \bar{\epsilon} \left(\mathds{1}_4 - \Gamma_5\right) \,.
\end{align}

\section{Null Reduction Results}\label{sec:NullReduction}

For the convenience of the reader, we review a few results on null reduction \cite{Julia:1994bs}. We follow the conventions of \cite{Bergshoeff:2017dqq}, which can also be consulted for more details.\\
The Ansatz \eqref{eq:vielbeinansatz} for the four-dimensional Vierbein $E_M{}^A$ in coordinates adapted to a null Killing vector $K^M$ leads to the following expressions for the components of the four-dimensional spin connection $\Omega_M{}^{AB}$:
\begin{alignat}{2}
  \label{eq:spinconnred}
  \Omega_{\mathsf{v}}{}^{+-} &= 0 \,, \qquad \qquad & \Omega_{\mathsf{v}}{}^{a-} &= 0 \,, \nonumber \\
  \Omega_{\mathsf{v}}{}^{a+} &= -\frac12 e^{\mu a} \tau^\rho \tau_{\mu\rho} \,, \qquad \qquad & \Omega_{\mathsf{v}}{}^{ab} &= \frac12 e^{\mu a} e^{\rho b} \tau_{\mu\rho} \,, \nonumber \\
  \Omega_\mu{}^{+-} &= -\frac12 \tau^\rho \tau_{\mu\rho} \,, \qquad \qquad & \Omega_\mu{}^{a-} &= \frac12 e^{\rho a} \tau_{\mu\rho} \,, \nonumber \\
  \Omega_\mu{}^{a+} &= -\omega_\mu{}^a + \frac12 m_\mu e^{\rho a} \tau^\sigma \tau_{\rho \sigma} \,, \qquad \qquad & \Omega_\mu{}^{ab} &= \omega_\mu{}^{ab} - \frac12 m_\mu e^{\rho a} e^{\sigma b} \tau_{\rho\sigma} \,,
\end{alignat}
where we have defined the Newton-Cartan spin connections $\omega_\mu{}^a$, $\omega_\mu{}^{ab}$ and Newton-Cartan torsion tensor $\tau_{\mu\nu}$ as follows:
\begin{align}
  \label{eq:defNCconnstaumunu}
  \omega_\mu{}^a &= e^{\nu a} \partial_{[\mu} m_{\nu]} - e_\mu{}^b e^{\nu a} \tau^\rho \partial_{[\nu} e_{\rho]b} - \tau^\nu \partial_{[\mu} e_{\nu]}{}^a - \tau_\mu e^{\nu a} \tau^\rho \partial_{[\nu} m_{\rho]} \,, \nonumber \\
  \omega_\mu{}^{ab} &= 2 e^{\nu[a} \partial_{[\mu} e_{\nu]}{}^{b]} - e_\mu{}^c e^{\nu a} e^{\rho b} \partial_{[\nu} e_{\rho] c} - \tau_\mu e^{\nu a} e^{\rho b} \partial_{[\nu} m_{\rho]} \,, \nonumber \\
  \tau_{\mu\nu} &= 2 \partial_{[\mu} \tau_{\nu]} \,.
\end{align}
Note that under Galilean boosts with parameter $\lambda^a$, the connections $\omega_\mu{}^a$ and $\omega_\mu{}^{ab}$ transform as follows
\begin{align} \label{eq:boosttrafoconns}
  \delta \omega_\mu{}^a &= - \partial_\mu \lambda^a + \lambda_b \omega_\mu{}^{ba} + \frac12 \lambda^b e_{\mu b} \tau^{a0} - \frac12 \lambda^a \tau_{\mu 0} \,, \nonumber \\
  \delta \omega_\mu{}^{ab} &= - \lambda^{[a} \tau_\mu{}^{b]} - \frac12 \lambda_c e_\mu{}^c \tau^{ab} \,.
\end{align}
We have the following expressions for the curvatures
\begin{align}
    R_{\mu\nu}(J) &= 2\partial_{[\mu}\omega_{\nu]} - 2 \omega_{[\mu}{}^a e_{\nu] a}\tau^{bc}\epsilon_{bc} - 2 \omega_{[\mu}{}^a\tau_{\nu]}\epsilon_{ab}\tau^{b0}\,,\label{eq:Jcurv}\\
    R_{\mu\nu}(G^a) &= 2\partial_{[\mu}\omega_{\nu]}{}^a + \epsilon^{ab}\omega_{[\mu}\omega_{\nu]}{}^b  + \omega_{[\mu}{}^{a} \tau_{\nu]}{}^{0}+ \omega_{[\mu}{}^{b} e_{\nu]}{}^b \tau^{a0}\,,\label{eq:Gcurv}
\end{align}
where we defined $\omega_\mu = \omega_\mu{}^{ab}\epsilon_{ab}$.

We also sometimes use an affine, torsionful connection $\bar\Gamma$ that is defined by the following Vielbein postulates
\begin{align}
    \bar\nabla_\mu \tau_\nu &= \partial_\mu\tau_\nu - \bar\Gamma_{\mu\nu}^\rho\tau_\rho = 0\,,\\
    \bar\nabla_\mu e_\nu{}^a &= \partial_\mu e_\nu{}^a + \omega_\mu{}^{ab}e_\nu{}^b + \omega_\mu{}^a\tau_\mu - \bar\Gamma_{\mu\nu}^\rho e_\rho{}^a = 0\,.
\end{align}
Using that the spin connections solve $2\partial_{[\mu}e_{\nu]}{}^a + 2\omega_{[\mu}{}^{ab}e_{\nu]}{}^b + 2\omega_{[\mu}{}^a\tau_{\nu]}=0$ and $2\partial_{[\mu}m_{\nu]} - 2\omega_{[\mu}{}^a e_{\nu]}{}^a=0$, one finds
\begin{align}
  \label{eq:NCaffconn}
  \bar{\Gamma}^\rho_{\mu\nu} &= \tau^\rho \partial_\mu \tau_\nu + \frac12 h^{\rho\sigma} \left( \partial_\mu h_{\sigma\nu} + \partial_\nu h_{\sigma \mu} - \partial_\sigma h_{\mu\nu} \right) + \tau_\mu h^{\rho\sigma} \partial_{[\nu} m_{\sigma]} + \tau_\nu h^{\rho\sigma} \partial_{[\mu} m_{\sigma]} \,,
\end{align}
where $h^{\mu\nu} = e^\mu{}_a e^\nu{}_a$. Observe that this affine Newton-Cartan connection has torsion $2\bar\Gamma^\rho_{[\mu\nu]} = \tau^\rho\tau_{\mu\nu}$. As a consequence
\begin{align}\label{eq:partialIntegration}
    \det(e^a,\tau)\bar\nabla_\mu X^\mu = \partial_\mu\lr \det(e^a,\tau)\,X^\mu\rr + \det(e^a,\tau)\tau_{0\mu}X^\mu\,.
\end{align}

We decompose the four-dimensional Clifford algebra matrices $\Gamma_A$ as tensor products of two $(2\times 2)$-matrices as follows:
\begin{align}
  \label{eq:Gammadecomp}
  \Gamma_\pm = \gamma_0 \otimes \sigma_\pm \,, \qquad \Gamma_a = \gamma_a \otimes \mathds{1}_2 \,, \quad a = 1,2 \,,
\end{align}
where $\sigma_\pm$ is given in terms of the Pauli-matrices $\sigma_1$ and $\sigma_2$ by
\begin{align}
  \label{eq:defsigmapm}
  \sigma_\pm = \frac{1}{\sqrt{2}} (\sigma_1 \pm \rmi \sigma_2 ) \,,
\end{align}
and $\gamma_0$, $\gamma_a$ are the gamma-matrices of a three-dimensional Clifford algebra, normalized as follows
\begin{equation}
    \gamma_0^2 = - \mathds{1}_2,\qquad \{\gamma_a, \gamma_b\} = 2\, \delta_{ab}\, \mathds{1}_2\qquad\text{and}\qquad \{\gamma_0,\gamma_a\}=0. \label{eq:3dClifford}
\end{equation}
The following gamma-matrix relations hold:
\begin{align}
  \gamma_{ab} = \epsilon_{ab} \gamma_0 \,, \qquad \qquad \gamma_{a0} = \epsilon_{ab} \gamma_b \,.
\end{align}
The four-dimensional charge conjugation matrix $C$ decomposes as
\begin{align}
  \label{eq:Cdecomp}
  C = \mathcal{C}_3 \otimes \sigma_1 \,,
\end{align}
where $\mathcal{C}_3$ is the three-dimensional charge conjugation matrix obeying
\begin{align}
  \label{eq:propC3}
  \mathcal{C}_3^T = - \mathcal{C}_3 \,, \qquad \gamma_0^T = - \mathcal{C}_3 \gamma_0 \mathcal{C}_3^{-1} \,, \qquad \gamma_a^T = - \mathcal{C}_3 \gamma_a \mathcal{C}_3^{-1} \,.
\end{align}

\section{Integrability Conditions}\label{sec:Integrability}

The Killing spinor equations described in Section \ref{ssec:KSeqs} lead to integrability conditions, which must be satisfied for every consistent solution. In this appendix, we perform an analysis of these conditions for each of the cases listed in Sections \ref{ssec:zetapluszero} and \ref{ssec:zetaplusneq0}. For convenience, we first organize the two differential Killing spinor equations in a way that highlights the three different gamma matrix structures that appear. Specifically, 
\begin{subequations}\label{kse1}
	\begin{align} 
&\bar\nabla_{\mu}\epsilon_{+}:=\left(P_{\mu}^{+}+Q_{\mu}^{+}\gamma_0+R_{\mu}^{a+}\gamma_a\right)\epsilon_{+}+\left(\widetilde{P}_{\mu}^{+}+\widetilde{Q}_{\mu}^{+}\gamma_0+\widetilde{R}_{\mu}^{a+}\gamma_a\right)\epsilon_{-}\,,\\[4pt]
&\bar{\nabla}_{\mu}\epsilon_{-}:=\left(P_{\mu}^{-}+Q_{\mu}^{-}\gamma_0+R_{\mu}^{a-}\gamma_a\right)\epsilon_{-}+\widetilde{R}_{\mu}^{a-}\gamma_{a}\epsilon_{+}\,,\label{kse2}
\end{align}
\end{subequations}
where the covariant derivatives are defined as in Eqs. \eqref{eq:covderepspm} and the tensors appearing in \eqref{kse1} are 
	\begin{align}
&P_{\m}^{+}=-\sfrac 14 \t_{\m}{}^{0}-\sfrac 16 e_{\m}{}^{a}v_{b}\e_{ab}\,, \quad \widetilde{P}_{\m}^+=-\sfrac{\sqrt{2}}{6}\text{Im}(u)\t_{\m}\,,\nn\\[4pt]
& Q_{\m}^{+}=-\sfrac 13 e_{\m}{}^{a}v_{a}-\sfrac 16 \t_{\m}v_0\,, \quad \widetilde{Q}_{\m}^{+}=-\sfrac{\sqrt{2}}{6}\text{Re}(u)\t_{\m}\,,\nn\\[4pt]
& R_{\m}^{a+}=-\sfrac 16 \text{Re}(u)e_{\m}{}^{a}+\sfrac 16 \text{Im}(u)e_{\m}{}^{b}\e_{ab}\,,\quad \widetilde{R}_{\m}^{a+}=-\sfrac{\sqrt{2}}{6}ve_{\m}{}^{a}-\sfrac{\sqrt{2}}{6}\tau_{\mu}v_a+\sfrac{\sqrt{2}}{4}\t_{\m}{}^{b}\e_{ab}\,,\nn\\[4pt]
&P_{\m}^{-}=\sfrac 14 \t_{\m}{}^{0}+\sfrac 16 e_{\m}{}^{a}v_{b}\e_{ab}\,,\quad Q_{\m}^{-}=\sfrac 13e_{\m}{}^{a}v_a+\sfrac 12\t_{\m}v_0\,,\nn\\[4pt]
& R_{\m}^{a-}=-\sfrac 16 \text{Re}(u)e_{\m}{}^{a}-\sfrac 16 \text{Im}(u)e_{\m}{}^{b}\e_{ab}\,,\quad \widetilde{R}_{\m}^{a-}=\sfrac{\sqrt{2}}{6}e_{\m}{}^{a}v_0\,.\label{ICtensors}
\end{align}
Next, since the covariant derivatives \eqref{eq:covderepspm} transform under boosts as 
\begin{align} 
&	\d(\bar\nabla_{\mu}\epsilon_{+})= -\frac 14 \l_ce_{\m}{}^{c}\tau^{ab}\g_{ab}\epsilon_+-\frac 14 \l^a\tau_{\mu}\tau_{0}{}^{b}\g_{ab}\epsilon_+\,, \\
& \d(\bar\nabla_{\mu}\epsilon_{-})=-\frac{\sqrt{2}}{2}\l^a\g_{a0}\bar\nabla_{\m}\epsilon_+-\frac 14 \l_ce_{\m}{}^{c}\tau^{ab}\g_{ab}\epsilon_--\frac 14 \l^a\tau_{\mu}\tau_{0}{}^{b}\g_{ab}\epsilon_--\nn\\
& \qquad\quad\qquad -\frac{\sqrt{2}}{4}\l^be_{\m b}\tau^{a0}\g_{a0}\epsilon_++\frac{\sqrt{2}}{4}\l^ae_{\m}{}^{b}\tau_{b0}\g_{a0}\epsilon_+\,,
	\end{align} 
we define 
\begin{align} 
&	\bar\nabla_{\m}\bar\nabla_{\n}\epsilon_+=(\partial_{\m}+\frac 14 \omega_{\m}{}^{ab}\g_{ab})\bar\nabla_{\n}\epsilon_+-\frac 14 \omega_{\m}{}^{c}e_{\n c}\tau^{ab}\g_{ab}\epsilon_+-\frac 14 \omega_{\m}{}^{a}\tau_{\n}\tau_{0}{}^{b}\g_{ab}\epsilon_+\,, 
	\\ 
& \bar\nabla_{\m}\bar\nabla_{\n}\epsilon_-=	(\partial_{\m}+\frac 14 \omega_{\m}{}^{ab}\g_{ab})\bar\nabla_{\n}\epsilon_--\frac{\sqrt{2}}{2}\omega_{\mu}{}^{a}\g_{a0}\bar\nabla_{\n}\epsilon_+-\frac 14 \omega_{\m c}e_{\n}{}^{c}\tau^{ab}\g_{ab}\epsilon_--\nn\\ 
&\qquad\quad\qquad -\frac 14 \omega_{\mu}{}^{a}\tau_{\nu}\tau_{0}{}^{b}\g_{ab}\epsilon_--\frac{\sqrt{2}}{4}\omega_{\m}{}^{b}e_{\nu b} \tau^{a0}\g_{a0}\epsilon_{+}+\frac{\sqrt{2}}{4}\omega_{\m}{}^{a}e_{\n}{}^{b}\tau_{b0}\g_{a0}\epsilon_+\,. 
	\end{align}
A straightforward computation gives the following result for the commutation relations of covariant derivatives, 
\begin{align}  
	& [\bar\nabla_{\m},\bar\nabla_{\n}]\epsilon_{+}=\frac 14 R_{\m\n}(J)\g_0\epsilon_+\,,\\
	& [\bar\nabla_{\m},\bar\nabla_{\n}]\epsilon_{-}=\frac 14 R_{\m\n}(J)\g_0\epsilon_-+\frac{\sqrt{2}}{2}\epsilon_{ab}R_{\m\n}(G^b)\g^a\epsilon_+\,,
	\end{align} 
with the curvatures given as in \eqref{eq:Jcurv} and \eqref{eq:Gcurv}. 
In the spirit of Section 4, we express the solutions of the Killing spinor equations  in terms of commuting spinors $(\zeta_+,\zeta_-)$ and perform an analysis for the following cases. 

\subsection[The case $\zeta_+ = 0$]{The case \boldmath${\zeta_+ = 0}$}

In the present case, the covariant derivatives as defined above satisfy the commutation relation 
\begin{equation} 
	[\bar{\nabla}_{\mu},\bar{\nabla}_{\nu}]\z_{-}=\frac 14 R_{\m\n}(J)\gamma_0\z_-\,.
	\end{equation}
Let us first assume that the spinor $\z_-$ is not further constrained and there are two solutions to the Killing spinor equations. This is the case when $u=0$, as discussed in the proof of Theorem \ref{thm1}. Then a direct calculation on the basis of \eqref{kse1} leads to the integrability condition 
\bea \label{IC1a}
\left(U_{\m\n}+V_{\m\n}\g_0+W_{\m\n}^{a}\g_a\right)\zeta_-=0\,,
\eea  
where 
\bse
 \begin{align}
& U_{\m\n}=2\partial_{[\m}P_{\n]}^{-}\,,\\[4pt] 
& V_{\m\n}=-\frac 14 R_{\m\n}(J)+2\partial_{[\m}Q_{\n]}^{-}-2\epsilon_{ab} R_{[\m}^{a-}R_{\n]}^{b-}-\frac 12 \omega_{[\mu}{}^{c}e_{\nu]c}\tau^{ab}\epsilon_{ab}-\frac 12 \omega_{[\m}{}^{a}\tau_{\nu]}\tau_{0}{}^{b}\epsilon_{ab}\,,\\[4pt]
& W_{\m\n}^{a}=2\partial_{[\m}R_{\nu]}^{a-}+\epsilon^{a}{}_{b}\omega_{[\m}R_{\n]}^{b-}-4\epsilon^{a}{}_{b}Q_{[\m}^{-}R_{\n]}^{b-}\,.
\end{align}
\ese
 Thus we are directly led to impose three conditions, namely 
\be 
U_{\m\n}=0\,,\quad V_{\m\n}=0\,,\quad W_{\m\n}^{a}=0\,.
\ee 
Upon substitution of $P_{\m}^{-}$ as given in \eqref{ICtensors}, a straightforward calculation shows that $U_{\m\n}=0$ is an identity due to the relation
\be 
\tau_{0a}=-\frac 23 \epsilon_{ab}v^b\,,
\ee 
which follows from the algebraic Killing spinor equation \eqref{eq:algepsp02} in the present case that $u=0$. 
In addition, $R_{\m}^{a-}$ vanishes and thus the condition $W_{\m\n}^{a}=0$ is identically satisfied as well. The remaining integrability condition, $V_{\m\n}=0$, can be algebraically manipulated and it yields the final result 
\be 
R_{\m\n}(J)=\frac 83 D_{[\m}v_{\n]}-\frac 43 \tau_{[\m}D_{\n]}v_0+\frac 23 \tau_{\m\n}v_0\,,
\ee 
where the boost covariant derivatives on $v_{\m}$ and $v_0$ are defined as 
\bea 
&& D_{\m}v_{\n}=\partial_{\m}v_{\n}-\omega_{\m}{}^{a}e_{\n a}v\,,\\[4pt]
&& D_{\m}v_0=\partial_{\m}v_0-\omega_{\m}{}^{a}v_a\,.
\eea 

The above situation changes when the spinor $\z_{-}$ is constrained further. Indeed, when it obeys 
\be 
X^{-a}\g_a\z_{-}=\z_{-}\,,
\ee 
it is simple to show that 
\be 
\gamma_a\z_-=Y^{-}_a\g_0\z_-+X_a^{-}\z_-\,,
\ee 
and therefore there is a reduction of the possible gamma matrix structures in the Killing spinor equations and in the corresponding integrability conditions. One should then be cautious and rederive the integrability conditions, since now the derivative can act on $X_a^-$ too and thus \eqref{IC1a} is no longer the correct condition.  Since now there are only two possible gamma matrix structures, there are two conditions which read as 
\bse\begin{align}
& 2\partial_{[\m}\left(P_{\n]}^-+R_{\n]}^{a-}X_a^{-}\right)=0\,,\label{IC2a}\\
& 2\partial_{[\m}\left(Q_{\n]}^-+R_{\n]}^{a-}Y_a^{-}\right)-\frac 12 \omega_{[\mu}{}^{c}e_{\nu]c}\tau^{ab}\epsilon_{ab}-\frac 12 \omega_{[\m}{}^{a}\tau_{\nu]}\tau_{0}{}^{b}\epsilon_{ab}=\frac 14 R_{\m\n}(J)\,.\label{IC2b}
\end{align}\ese
Substituting the tensors given in \eqref{ICtensors}, one finds that the first condition \eqref{IC2a} reads
\be 
\partial_{[\m}\left(e_{\n]}{}^{a}\left(\frac 32\tau_{0a}+\epsilon_{ab}v^b-\text{Re}(u)X_a^{-}+\text{Im}(u)Y_a^{-} \right)\right)=0\,,
\ee  
and therefore it is satisfied identically due to Theorem \ref{thm1}. The remaining integrability condition stemming from \eqref{IC2b} is then found to be 
\bea 
R_{\m\n}(J)&=&\frac 83 D_{[\mu}v_{\n]}-\frac 43 \tau_{[\mu}D_{\nu]}v_0+\frac 2 3\tau_{\m\n}v_0-\nn\\[4pt] && -\,\frac 4 3 e_{[\n}{}^aY_a^-\partial_{\m]}(\text{Re}{(u)})-\frac 43 e_{[\n}{}^{a}X_a^-\partial_{\m]}(\text{Im}(u))-\nn\\[4pt] 
&& -\,\frac 43e_{[\n}{}^{a}\text{Re}(u)D_{\m]}Y_a^-
-\frac 43 e_{[\n}{}^{a}\text{Im}(u)D_{\m]}X_a^-\,,\label{ICzeta-}
\eea 
where the covariant derivatives on $X_a^-$ and $Y_a^-$ are as defined in Theorem \ref{thm1}.

\subsection[The case $\zeta_+ \ne 0$]{The case \boldmath${\zeta_+ \ne 0}$}

Following the same logic as in the previous case, first we examine the integrability condition in case there are two solutions, which means that $u=0$, as shown in the main text. Then $R_{\mu}^{a+}=0$ and a straightforward computation leads to the two conditions 
\bse\begin{align} 
	& 2\partial_{[\mu}P_{\n]}^{+}=0\,,\\ 
	& -\frac 14 R_{\m\n}(J)+2\partial_{[\m}Q_{\n]}^{+}-\frac 12 \omega_{[\mu}{}^{c}e_{\nu]c}\tau^{ab}\epsilon_{ab}-\frac 12 \omega_{[\m}{}^{a}\tau_{\nu]}\tau_{0}{}^{b}\epsilon_{ab}=0\,.\label{case2acondition2}
	\end{align}  \ese
The first condition, using the explicit expression for $P_{\m}^+$, becomes 
\be 
\partial_{[\m}\left(e_{\n]}{}^a\left(\tau_{a}{}^{0}+\frac 23 \epsilon_{ab}v^{b}\right)\right)=0\,.
\ee
However, since in the present case (see Theorem \ref{thm2}) it holds that $\tau_{a}{}^{0}=\frac 23 \epsilon_{ab}v^{b}$ for $u=0$, we directly obtain that the 1-form $\tau_{\m 0}$ must be closed, 
\be \label{closure}
\partial_{[\mu}\tau_{\n]0}=0\,.
\ee 
This is in agreement with Theorem \ref{thm2}, where this 1-form is even exact, and therefore this condition does not pose further restrictions. On the other hand, the second condition \eqref{case2acondition2}, upon using the algebraic Killing spinor equations as in Theorem \ref{thm2}, becomes 
\be 
R_{\m\n}(J)= -\frac 83 D_{[\mu}v_{\n]}-\frac 43 \tau_{[\mu}D_{\nu]}v_0-\frac 29 \tau_{\m\n}v_0+\frac 49 e_{[\m}{}^{a}\tau_{\n]}v_0\tau_{a0}\,,
\ee 
with the covariant derivatives defined as before but specialized to the value of the boost connection $\omega_{\mu}{}^{a}$ given in \eqref{eq:solboostconn}.  

As in the previous case, when the spinor $\z_{+}$ is constrained further and $u\ne 0$, the integrability condition changes. The spinor obeys 
\be 
X^{+a}\g_a\z_{+}=\z_{+}\,,
\ee 
and therefore
\be 
\gamma_a\z_+=Y^{+}_a\g_0\z_++X_a^{+}\z_+\,,
\ee 
 which reduces the possible gamma matrix structures in the Killing spinor equations and in the corresponding integrability conditions. The resulting two conditions are analogous to the corresponding ones for $\zeta_+=0$ and they read as 
\bse\begin{align}
	& 2\partial_{[\m}\left(P_{\n]}^++R_{\n]}^{a+}X_a^{+}\right)=0\,,\label{IC2a+}\\
	& 2\partial_{[\m}\left(Q_{\n]}^++R_{\n]}^{a+}Y_a^{+}\right)-\frac 12 \omega_{[\mu}{}^{c}e_{\nu]c}\tau^{ab}\epsilon_{ab}-\frac 12 \omega_{[\m}{}^{a}\tau_{\nu]}\tau_{0}{}^{b}\epsilon_{ab}=\frac 14 R_{\m\n}(J)\,.\label{IC2b+}
\end{align}\ese
Upon substitution of $P_{\m}^{+}$ and $R_{\m}^{a+}$ from \eqref{ICtensors},  the first condition \eqref{IC2a} becomes 
\be 
\partial_{[\m}\left(e_{\n]}{}^{a}\left(\frac 32\tau_{0a}+\epsilon_{ab}v^b+\text{Re}(u)X_a^{+}+\text{Im}(u)Y_a^{+} \right)\right)=0\,,
\ee  
and combining it with the algebraic Killing spinor equation that leads to \eqref{eq:soltau0app} of Theorem \ref{thm2}, it translates once more to the closedness of the one-form $\tau_{\mu 0}$, namely to 
\eqref{closure}, which holds because the one-form is exact. In turn, the second condition \eqref{IC2b+} becomes 
\begin{align} 
& R_{\m\n}(J)= -\frac 83 D_{[\mu}v_{\n]}-\frac 43 \tau_{[\mu}D_{\nu]}v_0-\frac 29 \tau_{\m\n}v_0+\frac 49 e_{[\m}{}^{a}\tau_{\n]}v_0(\tau_{a0}-2\text{Re}(u)X_a^++2\text{Im}(u)Y_a^+) \, - \nn\\ 
& \qquad \quad\qquad - \frac 43 e_{[\nu}{}^{a}Y_a^+\partial_{\mu]}(\text{Re}(u))+\frac 43 e_{[\nu}{}^{a}X_a^+\partial_{\mu]}(\text{Im}(u)) \,- \nn\\ 
& \qquad\quad\qquad -\frac 43 e_{[\nu}{}^{a}\text{Re}(u)D_{\mu]}Y_a^++\frac 43 e_{[\nu}{}^a\text{Im}(u)D_{\mu]}X_a^+ \,,
\end{align} 
with the covariant derivatives on $X_a^+$ and $Y_a^+$ as in Theorem \ref{thm2}.

\providecommand{\href}[2]{#2}\begingroup\raggedright\endgroup


\end{document}